%% file: robust_compliance.tex
  \RequirePackage{fix-cm}
    \documentclass[twocolumn]{svjour3}          
  \smartqed  
  \usepackage{placeins}
  \usepackage{graphicx}
  \usepackage{appendix}
  \usepackage[english]{babel}
  \usepackage{verbatim}
  \usepackage{pifont}
  \usepackage[inline]{enumitem}
  \usepackage{graphicx}
  \usepackage{natbib}
  \usepackage{optidef}
  \usepackage{amsmath}
  \usepackage{amssymb}
  \usepackage{bm}
  \usepackage{algorithm, algorithmic}

  \usepackage{float}
  \usepackage{caption}
  \usepackage{subcaption}
  \usepackage{pgfplots}
  \usepackage{svg}
  \usepackage{multirow}
  \usepackage{tikzscale}
  \usetikzlibrary{calc}
  \usetikzlibrary{decorations,decorations.markings,decorations.text}
  \usetikzlibrary{patterns}
  \usetikzlibrary{arrows}
  \pgfdeclarepatternformonly{north east lines wide}{\pgfqpoint{-1pt}{-1pt}}{\pgfqpoint{10pt}{10pt}}{\pgfqpoint{9pt}{9pt}}%
  {
    \pgfsetlinewidth{0.4pt}
    \pgfpathmoveto{\pgfqpoint{0pt}{0pt}}
    \pgfpathlineto{\pgfqpoint{9.1pt}{9.1pt}}
    \pgfusepath{stroke}
  }
  \tikzset{%
      body/.style={inner sep=0pt,outer sep=0pt,shape=rectangle,draw,thick},
      marked body/.style={inner sep=0pt,outer sep=0pt,shape=rectangle,draw,pattern=north east lines},
      dimen/.style={<->,>=latex,thin,every rectangle node/.style={fill=white,midway,font=\sffamily}},
      symmetry/.style={dashed,thin},
  }
  
  \usepackage{makecell}
  \usepackage{hyperref}

\begin{document}

\title{Approximation schemes for stochastic compliance-based topology optimization with many loading scenarios}


\titlerunning{Stochastic compliance optimization}        

\author{Mohamed Tarek         \and
        Tapabrata Ray
}

\institute{Mohamed Tarek \at
              UNSW Canberra, Northcott Drive, Campbell, ACT 2600 \\
              \email{m.mohamed@student.adfa.edu.au}           
           \and
           Tapabrata Ray \at
              UNSW Canberra, Northcott Drive, Campbell, ACT 2600 \\
              \email{t.ray@adfa.edu.au}
}

\date{Received: date / Accepted: date}

\maketitle

\begin{abstract}
  In this paper, approximation schemes are proposed for handling load uncertainty in compliance-based topology optimization problems, where the uncertainty is described in the form of a set of finitely many loading scenarios. Efficient approximate methods are proposed to approximately evaluate and differentiate either 1) the mean compliance, or 2) a class of scalar-valued function of the individual load compliances such as the weighted sum of the mean and standard deviation. The computational time complexities of the proposed algorithms are analyzed, compared to the exact approaches and then experimentally verified. Finally, some mean compliance minimization problems and some risk-averse compliance minimization problems are solved for verification.

  \keywords{stochastic optimization \and risk-averse optimization \and compliance minimization \and SIMP \and  MMA \and FEA}

\end{abstract}

\section{Introduction}

  In practice, a topology optimization problem's data, e.g. the load applied or material properties, are typically uncertain and the optimal solution can often be sensitive to the specific values of such data, where a small change in some of the data can cause a significant change in the objective value or render the optimal solution obtained infeasible. There are multiple ways to model such data uncertainty. Robust optimization (RO), stochastic optimization (SO), risk-averse optimization (RAO) and reliability-based design optimization (RBDO) are some of the terms used in optimization literature to describe a plethora of techniques for handling uncertainty in the data of an optimization problem for different uncertainty models. In this paper, the focus will be on SO and RAO. For more on RO, the readers are refereed to \cite{Bertsimas2011} and \cite{AharonBen-Tal2009}. And for more on RBDO, the readers are referred to \cite{Choi2007} and \cite{Youn2004}.

  In SO and RAO, the data is assumed to follow a known probability distribution \citep{Shapiro2009,Choi2007}. Let $\bm{f}$ be a random load and $\bm{x}$ be the topology design variables. A probabilistic constraint can be defined as $P(g(\bm{x}; \bm{f}) \leq 0) \geq \eta$ where $\bm{f}$ follows a known probability distribution. This constraint is often called a chance constraint or a reliability constraint in RBDO. The objective of an SO problem is typically either deterministic or some probabilistic function such as the mean of a function of the random variable, its variance, standard deviation or a weighted sum of such terms. RAO can be considered a sub-field of SO borrowing concepts from risk analysis in mathematical economics to define various risk measures and tractable approximations to be used in objectives and/or constraints in SO. One such risk measure is the conditional value-at-risk (CVaR) \citep{Shapiro2009}. Other more traditional risk measures include the weighted sum of the mean and variance of a function or the weighted sum of the mean and standard deviation. For more on SO and RAO, the reader is referred to Shapiro et al. \cite{Shapiro2009}.

  In topology optimization literature, the term "\textit{robust topology optimization}" is often used to refer to minimizing the weighted sum of the mean, and variance or standard deviation of a function subject to probabilistic uncertainty \citep{Dunning2013,Zhao2014,Cuellar2018}. However, this use of the term "\textit{robust optimization}" is not consistent with the standard definition of RO in optimization theory literature, e.g. Ben-Tal et al. \cite{AharonBen-Tal2009}. The more compliant term to be used in this paper is \textit{stochastic topology optimization} or \textit{risk-averse topology optimization}.

  Many works in literature tackled the problem of load uncertainty in mean compliance minimization problems (\cite{Guest2008, Dunning2011, Zhao2014, Zhang2017a, Hutchinson1990, Liu2018a,tarek2021robust}). For a more detailed account of the literature on this, the readers are referred to \cite{tarek2021robust}. Of all the works reviewed, only two works (\cite{Zhang2017a,tarek2021robust}) dealt with data-driven design with a finite number of loading scenarios. The loading scenarios can be data collected or sampled from the distributions. The main limitation of the work by \cite{Zhang2017a} is that it can only be used to minimize the mean compliance which is not risk-averse since at the optimal solution, the compliance can still be very high for some probable load scenarios even if the mean compliance is minimized. In the work by (\cite{tarek2021robust}), a few exact methods for handling a large number of loading scenarios in compliance-based problems were proposed based on the singular value decomposition (SVD), where the loading matrix $\bm{F}$ has a low rank and/or only a few degrees of freedom are loaded. However when these conditions are not satisfied, the SVD based approach may not be efficient enough. In particular, there are 2 limitations to the SVD-based approaches:
  \begin{enumerate}
    \item The computational time complexity of computing the SVD of $\bm{F}$ is \newline $O(min(L, n_{dofs})^2 max(L, n_{dofs}))$ if the loads are dense, where $L$ is the number of loading scenarios and $n_{dofs}$ is the number of degrees of freedom, which can be computationally prohibitive for large problems.
    \item The load matrix may not be low rank.
  \end{enumerate}

  Some authors also studied risk-averse compliance minimization by considering the weighted sum of the mean and variance, the weighted sum of the mean and standard deviation, as well as other risk measures. (\cite{Dunning2013, Zhao2014b, Chen2010, Martinez-Frutos2016, Cuellar2018, Martinez-Frutos2018, Garcia-Lopez2013, Kriegesmann2019}). For a more detailed account of the literature on this, the readers are referred to \cite{tarek2021robust}.

  In this paper, a few computationally efficient, SVD-free approximation schemes are proposed to estimate the value and gradient of:
  \begin{enumerate}
      \item The mean compliance
      \item A class of scalar-valued functions of load compliances satisfying a few conditions, e.g. the standard deviation of the compliance
  \end{enumerate}
  subject to a finite number of possible loading scenarios. These approaches can be used in risk-averse compliance minimization. The rest of this paper is organized as follows. The proposed approaches for handling load uncertainty in continuum compliance problems in the form of a large, finite number of loading scenarios are detailed in sections \ref{sec:proposed_mean} and \ref{sec:proposed_risk}. The experiments used and the implementations are then described in section \ref{sec:exp_impl}. Finally, the results are presented and discussed in section \ref{sec:results} before concluding in section \ref{sec:conclusion}.

\section{Proposed algorithms}

  \subsection{Solid isotropic material with penalization}

    In this paper, the solid isotropic material with penalization (SIMP) method \citep{Bendsoe1989,Sigmund2001,Rojas-Labanda2015} is used to solve the topology optimization problems. Let $0 \leq x_e \leq 1$ be the decision variable associated with element $e$ in the ground mesh and $\bm{x}$ be the vector of such decision variables. Let $\rho_e$ be the pseudo-density of element $e$, and $\bm{\rho}(\bm{x})$ be the vector of such variables after sequentially applying to $\bm{x}$:
    \begin{enumerate}
      \item A chequerboard density filter typically of the form $f_1(\bm{x}) = \bm{A} \bm{x}$ for some constant matrix $\bm{A}$ \citep{Bendsoe2004}, 
      \item An interpolation of the form $f_2(y) = (1 - x_{min})y + x_{min}$ applied element-wise for some small $x_{min} > 0$ such as $0.001$, 
      \item A penalty such as the power penalty $f_3(z) = z^p$ applied element-wise for some penalty value $p$, and
      \item A projection method such as the regularized Heaviside projection \citep{Guest2004} applied element-wise.
    \end{enumerate}
    The compliance of the discretized design is defined as: $C = \bm{u}^T\bm{K}\bm{u} = \bm{f}^T\bm{K}^{-1}\bm{f}$ where $\bm{K}$ is the stiffness matrix, $\bm{f}$ is the load vector, and $\bm{u} = \bm{K}^{-1}\bm{f}$ is the displacement vector. The relationship between the global and element stiffness matrices is given by $\bm{K} = \sum\limits_e \rho_e \bm{K}_e$ where $\bm{K}_e$ is the hyper-sparse element stiffness matrix of element $e$ with the same size as $\bm{K}$.

  \subsection{Approximating the compliance sample mean and its gradient} \label{sec:proposed_mean}

    The mean compliance can be formulated as a trace function: $\mu_C = \frac{1}{L} tr(\bm{F}^T \bm{K}^{-1} \bm{F})$. \cite{Zhang2017a} showed that Hutchinson's trace estimator \cite{Hutchinson1990} can be used to accurately estimate the compliance for a large number of load scenarios using a relatively small number of linear system solves. Hutchinson's trace estimator is given by:
      \begin{align}
        & tr(\bm{A}) = E(\bm{v}^T \bm{A} \bm{v}) \approx \frac{1}{N} \sum_{i=1}^{N} \bm{v}_i^T \bm{A} \bm{v}_i
      \end{align}
      where $\bm{v}$ is a random vector with each element independently distributed with 0 mean and unit variance, $\bm{v}_i$ are samples of the random vector $\bm{v}$, also known as probing vectors, and $N$ is the number of such probing vectors. One common distribution used for the elements of $\bm{v}$ is the Rademacher distribution which is a discrete distribution with support $\{-1, 1\}$ each of which has a probability of 0.5. Hutchinson proved that an estimator with the Rademacher distribution for $\bm{v}$ will have the least variance among all other distributions. Let $\bm{A} = \bm{F}^T \bm{K}^{-1} \bm{F}$. The number of linear system solves required to compute the mean compliance $\frac{1}{L} tr(\bm{A})$ using the naive approach is $L$. However, when using Hutchinson's estimator, that number becomes the number of probing vectors $N$. In general, a good accuracy can be obtained for $N \ll L$. Other than the linear system solves, the time complexity of the remaining work using the trace estimation method is $O(N \times n_{dofs} \times L)$ mostly spent on finding $\bm{F} \bm{v}_i$ for all $i$. If only a small number of degrees of freedom $n_{loaded}$ are loaded, the complexity of the remaining work reduces to $O(N \times n_{loaded} \times L)$.

      Let $\bm{z}_i = \bm{K}^{-1} \bm{F} \bm{v}_i$ be cached from the trace computation. The elements of the gradient of the trace estimate with respect to $\bm{\rho}$ are given by:
      \begin{align}
       \mu_C(\bm{\rho}) & = \frac{1}{L \times N} \sum_i^N \bm{z}_i^T \bm{K} \bm{z}_i \\
       \frac{\partial \mu_C}{\partial \rho_e} & = \frac{1}{L \times N} \sum_{i=1}^N -\bm{z}_i^T \bm{K}_e \bm{z}_i
      \end{align}
      The additional time complexity of computing the gradient of the trace estimate after computing the trace is therefore $O(N \times n_E)$. For a detailed derivation of the partial above, see the appendix.

      \begin{table*}
       \centering
       \caption{Summary of the computational cost of the algorithms discussed to calculate the mean compliance and its gradient. \#Lin is the number of linear system solves required.}
       \begin{tabular}{| m{3cm} | m{0.7cm} | m{4cm} | m{5cm}|} 
        \hline
        \multirow{2}{3em}{Method} & \multirow{2}{2em}{\#Lin} & \multicolumn{2}{c|}{Time complexity of additional work} \\\cline{3-4}
        & & Dense & Sparse \\
        \hline
        \hline
        Exact & \(L\) & \(O(L \times (n_{dofs} + n_E))\) & \(O(L \times (n_{loaded} + n_E))\) \\
        \hline
        Trace estimation & \(N\) & \(O(N \times (n_{dofs} \times L + n_E))\) & \(O(N \times (L \times n_{loaded} + n_E))\) \\
        \hline
       \end{tabular}
       \label{tab:perf_mean}
      \end{table*}

  \subsection{Approximating scalar-valued function of load compliances and its gradient} \label{sec:proposed_risk}

    The above scheme for approximating the sample mean compliance can be generalized to handle the sample variance and standard deviations. The sample variance of the compliance $C$ is given by $\sigma_C^2 = \frac{1}{L-1} \sum_{i=1}^L (C_i - \mu_C)^2$. The sample standard deviation $\sigma_C$ is the square root of the variance. Let $\bm{C}$ be the vector of compliances $C_i$, one for each load scenario. In vector form, $\sigma_C^2 = \frac{1}{L-1} (\bm{C} - \mu_C \bm{1})^T (\bm{C} - \mu_C \bm{1})$. $\bm{C} = diag(\bm{A})$ is the diagonal of the matrix $\bm{A} = \bm{F}^T \bm{K}^{-1} \bm{F}$.

    One can view the load compliances $\bm{C}$ as the diagonal of the matrix $\bm{F}^T \bm{K}^{-1} \bm{F}$. One way to estimate it is therefore to use a diagonal estimation method. One diagonal estimator directly related to Hutchinson's trace estimator was proposed by Bekas et al. \cite{Bekas2007}. The diagonal estimator can be written as follows:
      \begin{align}
       & diag(\bm{A}) = E(\bm{D}_{\bm{v}} \bm{A} \bm{v}) \approx \frac{1}{N} \sum_{i=1}^N \bm{D}_{\bm{v}_i} \bm{A} \bm{v}_i
      \end{align}
    where $diag(\bm{A})$ is the diagonal of $\bm{A}$ as a vector, $\bm{D}_{\bm{v}}$ is the diagonal matrix with a diagonal $\bm{v}$, $\bm{v}$ is a random vector distributed much like in Hutchinson's estimator, $\bm{v}_i$ are the probing vector instances of $\bm{v}$ and $N$ is the number of probing vectors. The sum of the elements of the diagonal estimator above gives us Hutchinson's trace estimator. Let $\bm{A} = \bm{F}^T \bm{K}^{-1} \bm{F}$:
      \begin{align}
       & \bm{C} = diag(\bm{\bm{F}^T \bm{K}^{-1} \bm{F}}) \approx \frac{1}{N} \sum_{i=1}^N \bm{D}_{\bm{v}_i} \bm{\bm{F}^T \bm{K}^{-1} \bm{F}} \bm{v}_i
      \end{align}
    Bekas et al. showed that using the deterministic basis of a Hadamard matrix as probing vectors $\bm{v}_i$ rather than random vectors increases the accuracy of the diagonal estimator. In this paper, we do the same and use columns of a Hadamard matrix as probing vectors for the diagonal estimator. Given the diagonal estimate assuming $N \ll L$, one can estimate $\bm{C}$ using $N$ linear system solves, which can then be used to compute the sample variance and standard deviation. Other than the linear system solves, the additional work required above has a time complexity of $O(N \times n_{dofs} \times L)$. But if only a few $n_{loaded}$ degrees of freedom are loaded, the time complexity of the remaining work goes down to $O(N \times n_{loaded} \times L)$.

    The Jacobian of the compliances vector $\bm{C}$ with respect to $\bm{\rho}$, $\nabla_{\bm{\rho}} \bm{C}$ is simply the stacking of the transposes of the gradients of $C_i$, $\nabla_{\bm{\rho}} C_i$, for all $i$ to form a matrix with $L$ rows and $n_E$ columns. The Jacobian of the estimate of $\bm{C}$ is given by:
      \begin{align}
       & \nabla_{\bm{\rho}} \bm{C} \approx \frac{1}{N} \sum_{i=1}^N \bm{D}_{\bm{v}_i} \bm{\bm{F}}^T \bm{K}^{-1} \bm{F} \bm{v}_i = \frac{1}{N} \sum_{i=1}^N \bm{D}_{\bm{v}_i} \bm{F}^T \nabla_{\bm{\rho}} \bm{t}_i
      \end{align}
    where $\bm{t}_i = \bm{K}^{-1} \bm{F} \bm{v}_i$. The derivative $\frac{\partial \bm{t}_i}{\partial \rho_e} = -\bm{K}^{-1} \bm{K}_e \bm{t}_i$. Therefore:
      \begin{align}
       & \frac{\partial \bm{C}}{\partial \rho_e} \approx \frac{1}{N} \sum_{i=1}^N \bm{D}_{\bm{v}_i} \bm{\bm{F}}^T \bm{K}^{-1} \bm{K}_e \bm{K}^{-1} \bm{\bm{F}} \bm{v}_i
      \end{align}
    Note that to find the Jacobian in this case, one requires $L$ linear system solves to find $\bm{K}^{-1} \bm{F}$. However, with this many linear system solves one can use the exact method so there is no merit to using the diagonal estimation approach. This means that if the full Jacobian is required, the diagonal estimation method here is the wrong choice.

    However if only interested in $\nabla_{\bm{\rho}} \bm{C}(\bm{\rho})^T \bm{w}$, a more efficient approach can be used:
      \begin{align}
       \nabla_{\bm{\rho}} \bm{C}(\bm{\rho})^T \bm{w} & = \nabla_{\bm{\rho}} (\bm{C}(\bm{\rho})^T \bm{w}) = \nabla_{\bm{\rho}} tr(\bm{D}_{\bm{w}} \bm{F}^T \bm{K}^{-1} \bm{F})
      \end{align}
    Let $\bm{r}_i = \bm{K}^{-1} \bm{F} \bm{v}_i$ which are cached from the function value calculation, and let $\bm{t}_i = \bm{K}^{-1} \bm{F} \bm{D}_{\bm{w}} \bm{v}_i$.
      \begin{align}
       \frac{\partial \bm{C}(\bm{\rho})^T \bm{w}}{\partial \rho_e} & = -tr(\bm{D}_{\bm{w}} \bm{F}^T \bm{K}^{-1} \bm{K}_e \bm{K}^{-1} \bm{F}) \\
       & \approx -\frac{1}{N} \sum_{i=1}^N \bm{v}_i^T \bm{D}_{\bm{w}} \bm{F}^T \bm{K}^{-1} \bm{K}_e \bm{K}^{-1} \bm{F} \bm{v}_i \\
       & = -\frac{1}{N} \sum_{i=1}^N \bm{t}_i^T \bm{K}_e \bm{r}_i 
      \end{align}
    This means that at a cost of an additional $N$ linear system solves, one can compute the vectors $\bm{t}_i$ and then find the gradient of $\bm{C}^T \bm{w}$. Other than the linear system solves, the remaining work has a time complexity of $O(N \times (n_{dofs} \times L + n_E))$, $O(N \times n_{dofs} \times L)$ from the accumulation of $\bm{F} \bm{D}_{\bm{w}} \bm{v}_i$ and $O(N \times n_E)$ to evaluate the gradient given $\bm{t}_i$ and $\bm{r}_i$ for all $i$. If only a few degrees of freedom $n_{loaded}$ are loaded, then the complexity goes down to $O(N \times (n_{loaded} \times L + n_E))$.

    \begin{table*}
        \centering
        \caption{Summary of the computational cost of the algorithms discussed to calculate the load compliances $\bm{C}$ as well as $\nabla_{\bm{\rho}} \bm{C}^T \bm{w}$ for any vector $\bm{w}$. \#Lin is the number of linear system solves required. This can be used to compute the variance, standard deviation as well as other scalar-valued functions of $\bm{C}$.}
        \begin{tabular}{|m{4cm} | m{0.7cm} | m{4cm} | m{5cm} |} 
         \hline
         \multirow{2}{3em}{Method} & \multirow{2}{2em}{\#Lin} & \multicolumn{2}{c|}{Time complexity of additional work} \\\cline{3-4}
         & & Dense & Sparse \\
         \hline
         \hline
         Exact & \(L\) & \(O(L \times (n_{dofs} + n_E))\) & \(O(L \times (n_{loaded} + n_E))\) \\
         \hline
         Diagonal estimation & \(2N\) & \(O(N \times (n_{dofs} \times L + n_E))\) & \(O(N \times (n_{loaded} \times L + n_E))\) \\
        \hline
       \end{tabular}
       \label{tab:perf_scalar}
    \end{table*}

\section{Setup and Implementation} \label{sec:exp_impl}

  In this section, the most important implementation details and algorithm settings used in the experiments are presented.

  \subsection{Test problems}

    \input{./figures/cantbeam.tex}

    The 2D cantilever beam problem shown in Figure \ref{fig:CantBeam} was used to run the experiments. A ground mesh of plane stress quadrilateral elements was used, where each element is a square of side length $1 \text{ mm}$, and a sheet thickness of $1 \text{ mm}$. Linear iso-parametric interpolation functions were used for the field and geometric basis functions. A Young's modulus of 1 MPa and Poisson's ratio of 0.3 were used. Finally, a chequerboard density filter for unstructured meshes was used with a radius of 2 mm \citep{Huang2010a}. A 3D version of the problem above was also solved.

    Three variants of the cantilever beam problem were solved:
    \begin{enumerate}
      \item Minimization of the mean compliance $\mu_C$ subject to a volume constraint with a volume fraction of 0.4,
      \item Minimization of a weighted sum of the mean and standard deviation (mean-std) of the compliance $\mu_C + 2.0 \sigma_C$ subject to a volume constraint with a volume fraction of 0.4, and
      \item Volume minimization subject to a maximum compliance constraint with a compliance threshold of $70000 \text{ Nmm}$.
    \end{enumerate}
    A total of 1000 load scenarios were sampled from:
    \begin{align}
      \bm{f}_i = s_1 \bm{F}_1 + s_2 \bm{F}_2 + s_3 \bm{F}_3 + \frac{1}{R - 3} \sum_{j=4}^{R} s_j \bm{F}_j
    \end{align}
    where $\bm{F}_1$, $\bm{F}_2$ and $\bm{F}_3$ are unit vectors with directions as shown in Figure \ref{fig:CantBeam} and $R$ is an integer greater than or equal to 4. $\bm{F}_2$ and $\bm{F}_3$ are at 45 degrees. $s_1$, $s_2$ and $s_3$ are identically and independently uniformly distributed random variables between -2 and 2. $\bm{F}_j$ for $j$ in $4 \dots R$ are vectors with non-zeros at all the surface degrees of freedom without a Dirichlet boundary condition. The non-zero values are identically and independently normally distributed random variables with mean 0 and standard deviation 1. $s_j$ for $j$ in $4 \dots R$ are also identically and independently normally distributed random variables with mean 0 and standard deviation 1. The same loading scenarios were used for the 3 test problems. Let $\bm{F}$ be the matrix whose columns are the sampled $\bm{f}_i$ vectors. Given the way the loading scenarios have been defined the rank of $\bm{F}$ is almost certainly going to be around $R$.

  \subsection{Software}

    All the topology optimization algorithms described in this paper were implemented in TopOpt.jl \footnote{https://github.com/JuliaTopOpt/TopOpt.jl} using the Julia programming language \citep{Bezanson2014} for handling generic unstructured, iso-parametric meshes.

  \subsection{Settings}

    The value of $x_{min}$ used was $0.001$ for all problems and algorithms. Penalization was done prior to interpolation to calculate $\bm{\rho}$ from $\bm{x}$. A power penalty function and a regularized Heaviside projection were used. All of the problems were solved using 2 continuation SIMP routines. The first incremented the penalty value from $p = 1$ to $p = 6$ in increments of 0.5. Then the Heaviside projection parameter $\beta$ was incremented from $\beta = 0$ to $\beta = 20$ in increments of 4 keeping the penalty value fixed at 6. An exponentially decreasing tolerance from $1e-3$ to $1e-4$ was used for both continuations. 

    The mean and mean-std compliance minimization SIMP subproblems problems were solved using the method of moving asymptotes (MMA) algorithm \citep{Svanberg1987}. MMA parameters of $s_{init} = 0.5$, $s_{incr} = 1.1$ and $s_{decr} = 0.7$ were used as defined in the MMA paper with a maximum of 1000 iterations for each subproblem. The dual problem of the convex approximation was solved using a log-barrier box-constrained nonlinear optimization solver, where the barrier problem was solved using the nonlinear CG algorithm for unconstrained nonlinear optimization \citep{Nocedal2006} as implemented in Optim.jl \footnote{https://github.com/JuliaNLSolvers/Optim.jl} \citep{KMogensen2018}. The nonlinear CG itself used the line search algorithm from \cite{Hager2006} as implemented in LineSearches.jl \footnote{https://github.com/JuliaNLSolvers/LineSearches.jl}. The stopping criteria used was the one adopted by the KKT solver, IPOPT \citep{Wachter2006}. This stopping criteria is less scale sensitive than the KKT residual as it scales down the residual by a value proportional to the mean absolute value of the Lagrangian multipliers.

\section{Accuracy and speed comparison}

  In this section, the accuracy and speed of the approximations proposed are presented and compared to the exact values. A method to boost the accuracy of approximations is also presented and mathematically analyzed. Tables \ref{tab:time_mean} and \ref{tab:time_std} show the values computed for the mean compliance $\mu_C$ and its standard deviation $\sigma_C$ respectively together with the time required to compute their values and gradients using: the naive exact approach and the approximate method with trace or diagonal estimation using 100 Rademacher-distributed or Hadamard basis probing vectors. A value of $R = 10$ was used.
  \begin{table*}
   \centering
   \caption{The table shows the function values of $\mu_C$ computed using the exact method and the approximate method of trace estimation with 100 Rademacher-distributed or Hadamard basis probing vectors for a full ground mesh design. The table also shows the time required to compute or approximate $\mu_C$ and its gradient in each case. A value of $R = 10$ was used here.}
   \begin{tabular}{|c|c|c|}
    \hline
    Method & $\mu_C$ (Nmm) & Time (s) \\
    \hline
    \hline
    Naive Exact & 3328.7 & 24.2 \\
    \hline
    Trace estimation & 3596.9 (Rademacher) / 3486.7 (Hadamard) & 2.6 \\
    \hline
   \end{tabular}
   \label{tab:time_mean}
  \end{table*}

  \begin{table*}
   \centering
   \caption{The table shows the function values of $\sigma_C$ and its gradients for a full ground mesh computed using the exact method and the approximate method of diagonal estimation with 100 Rademacher-distributed or Hadamard basis probing vectors. The table also shows the time required to compute the exact or approximate $\sigma_C$ and its gradient in each case. A value of $R = 10$ was used here. Note the extreme bias in the estimate so a correction step is necessary.}
   \begin{tabular}{|c|c|c|}
    \hline
    Method & $\sigma_C$ (Nmm) & Time (s) \\
    \hline
    \hline
    Exact & 4172.8 & 28.0 \\
    \hline
    Diagonal estimation & 9774.8 (Rademacher) / 10173.3 (Hadamard) & 5.2 \\
    \hline
   \end{tabular}
   \label{tab:time_std}
  \end{table*}

  \begin{figure*}
    \begin{subfigure}[t]{0.45\textwidth}
      \centering
      \includegraphics[width=1\textwidth]{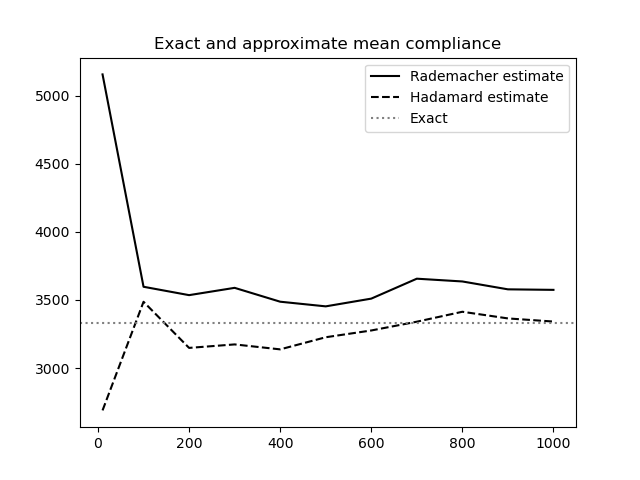}
      \caption{Mean compliance estimate using different numbers of probing vectors in the trace estimation method.}
      \label{fig:exact_approx_mean}
    \end{subfigure} \hfill
    \begin{subfigure}[t]{0.45\textwidth}
      \centering
      \includegraphics[width=1\textwidth]{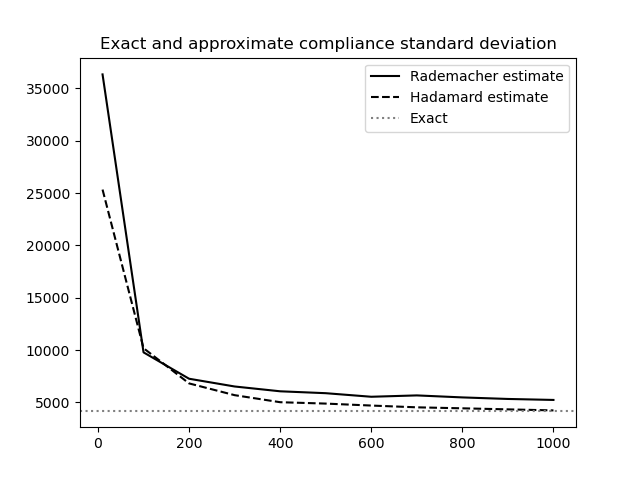}
      \caption{Compliance standard deviation estimate using different numbers of probing vectors in the diagonal estimation method.}
      \label{fig:exact_approx_std}
    \end{subfigure}
    \caption{Accuracy profile of the trace and diagonal estimation methods for estimating the mean compliance and its standard deviation using 10, 100, 200, 300, 400, 500, 600, 700, 800, 900 and 1000 probing vectors. A value of $R = 10$ was used here.}
    \label{fig:exact_approx}
  \end{figure*}

  As expected, the proposed approximation methods take a fraction of the time it takes to compute the exact mean and mean-std compliances using the approaches. Estimates of the mean compliance and its standard deviation for a full ground mesh using different numbers of Rademacher-distributed and Hadamard basis probing vectors are shown in Figure \ref{fig:exact_approx}. In this case, the estimates obtained using the Hadamard basis were always closer to the exact value than that of the Rademacher-distributed one. However, this depends on the order by which the Hadamard basis vectors are used.

\section{Bias correction}

  While the Hadamard estimate is converging faster to the exact value compared to the Rademacher one in the above case as the number of probing vectors increases, it is still quite far in the case of the standard deviation unless a large number ofIf we have a constraint over the weighted sum of the mean compliance and its standard deviation, this huge discrepancy from the exact quantity renders the approximate method useless. In this section, it will be shown that usually the estimate can be multiplied by a correcting factor to significantly improve its accuracy. This will be demonstrated experimentally and then mathematically analyzed. When performing topology optimization, the function value and its gradient need to be computed repeatedly. So if we only need to compute the correcting factor a few times, we can still save a lot of computational time when using the approximate method without losing too much accuracy.

\subsection{Experiments}

  Using a full ground mesh to calculate the correcting factor and only 10 probing vectors and $R = 10$, the ratio between the exact mean compliance and the trace estimate using Hadamard basis probing vectors was 1.238. Similarly, the ratio of the exact compliance standard deviation to the estimated value was 0.165. Figures \ref{fig:correcting_mean} and \ref{fig:correcting_std} show the distributions of the ratios of the exact value to the estimated one for the mean compliance and its standard deviation respectively. The same 10 Hadamard basis probing vectors were used and each figure was generated using 500 random designs. For each figure, the random designs were sampled from a truncated normal distributions with a different mean and a standard deviation of 0.2, truncated between 0 and 1. One can see that using the same probing vectors, the ratio between the exact and estimated values doesn't change significantly even when changing the mean volume by changing the mean of the truncated normal distribution. One can see that the correcting ratio that can multiply the estimated mean compliance or standard deviation to get the exact one is not very sensitive to the underlying design.

  \begin{figure*}[!htbp]
    \begin{subfigure}{0.3\textwidth}
      \centering
      \includegraphics[width=1\textwidth]{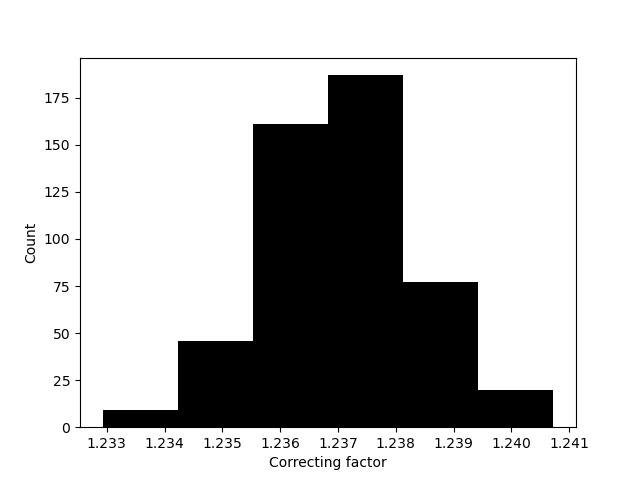}
      \caption{Mean = 0.1}
    \end{subfigure}
    \begin{subfigure}{0.3\textwidth}
      \centering
      \includegraphics[width=1\textwidth]{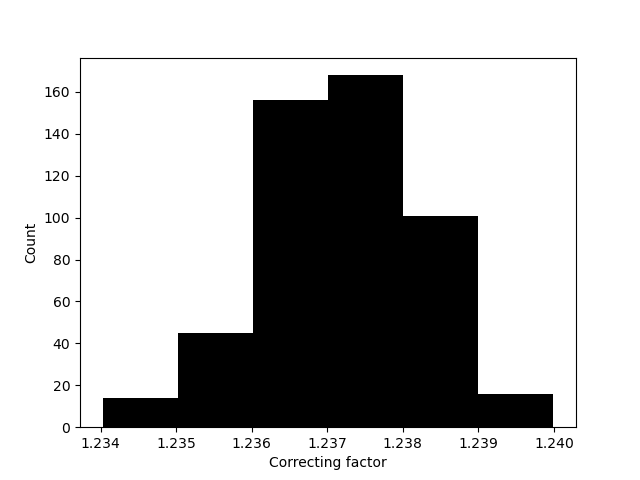}
      \caption{Mean = 0.3}
    \end{subfigure}
    \begin{subfigure}{0.3\textwidth}
      \centering
      \includegraphics[width=1\textwidth]{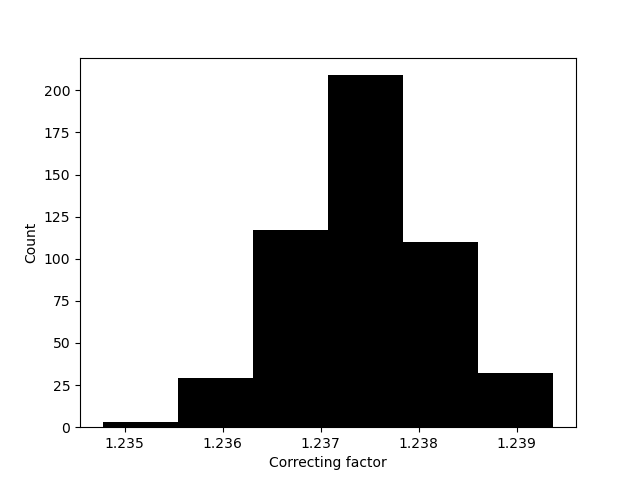}
      \caption{Mean = 0.5}
    \end{subfigure} 
    \begin{subfigure}{0.3\textwidth}
      \centering
      \includegraphics[width=1\textwidth]{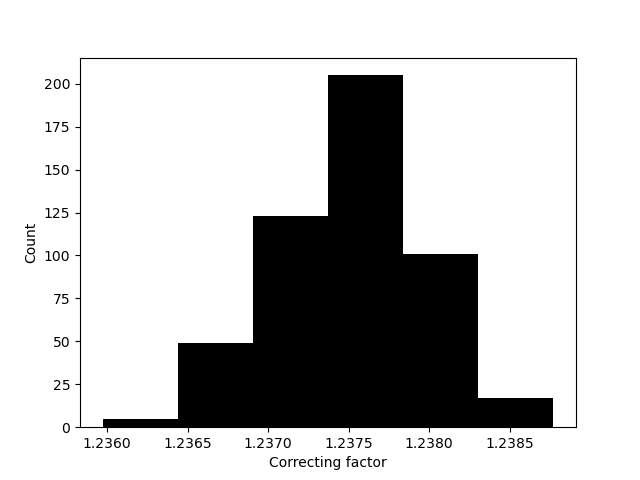}
      \caption{Mean = 0.7}
    \end{subfigure} 
    \begin{subfigure}{0.3\textwidth}
      \centering
      \includegraphics[width=1\textwidth]{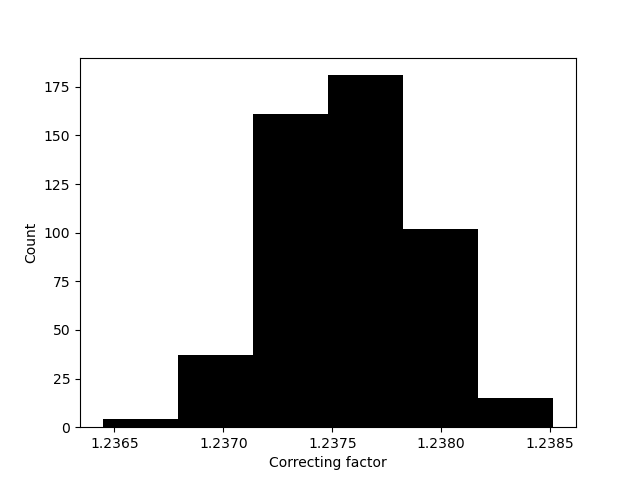}
      \caption{Mean = 0.9}
    \end{subfigure}
    \caption{Histograms of the ratio between the exact mean compliance and the trace estimate using 10 Hadamard basis probing vectors. In each figure, 500 designs were randomly sampled where each element's pseudo-density is sampled from a truncated normal distribution with the means indicated above and a standard deviation of 0.2, truncated between 0 and 1. A value of $R = 10$ was used here.}
    \label{fig:correcting_mean}
  \end{figure*}
  \begin{figure*}[!htbp]
    \begin{subfigure}[t]{0.3\textwidth}
      \centering
      \includegraphics[width=1\textwidth]{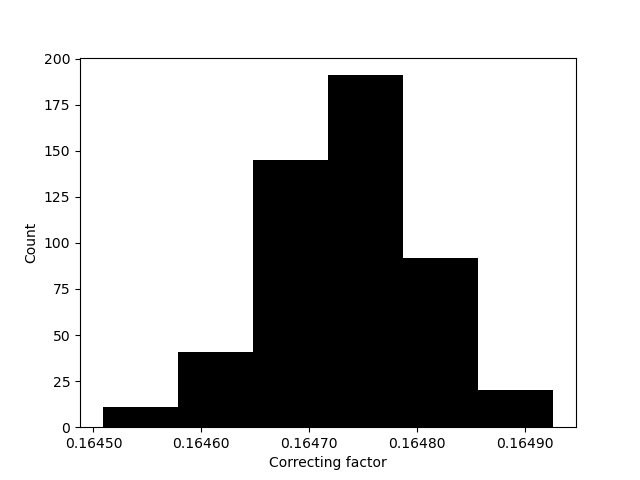}
      \caption{Mean = 0.1}
    \end{subfigure}
    \begin{subfigure}[t]{0.3\textwidth}
      \centering
      \includegraphics[width=1\textwidth]{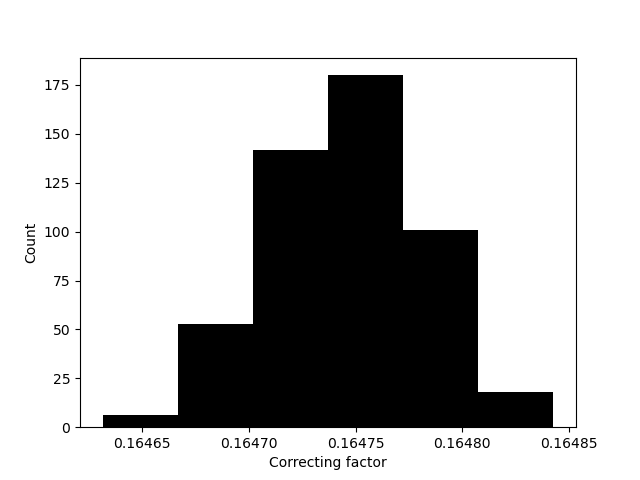}
      \caption{Mean = 0.3}
    \end{subfigure}
    \begin{subfigure}[t]{0.3\textwidth}
      \centering
      \includegraphics[width=1\textwidth]{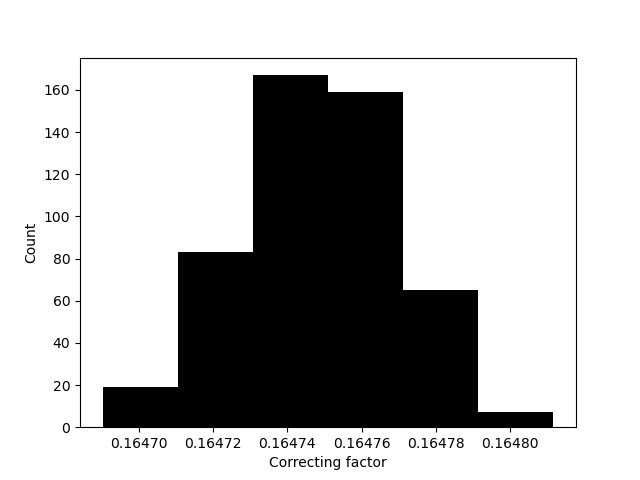}
      \caption{Mean = 0.5}
    \end{subfigure}
    \begin{subfigure}[t]{0.3\textwidth}
      \centering
      \includegraphics[width=1\textwidth]{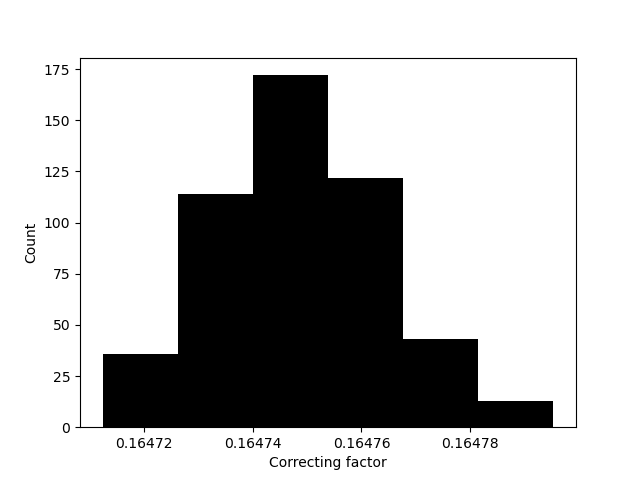}
      \caption{Mean = 0.7}
    \end{subfigure}
    \begin{subfigure}[t]{0.3\textwidth}
      \centering
      \includegraphics[width=1\textwidth]{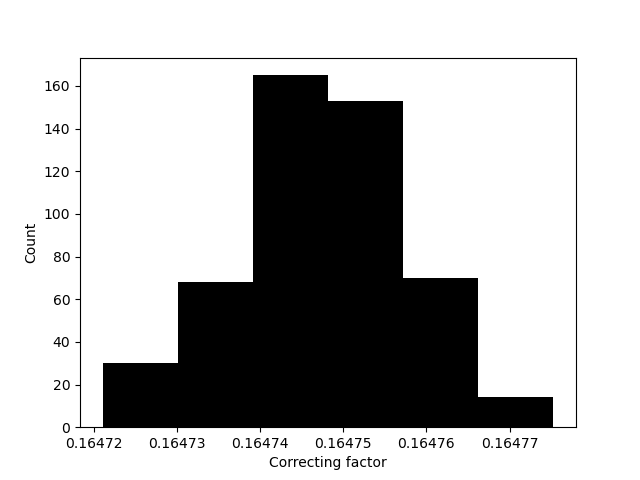}
      \caption{Mean = 0.9}
    \end{subfigure}
    \caption{Histograms of the ratio between the exact compliance standard deviation and the estimate using 10 Hadamard basis probing vectors. In each figure, 500 designs were randomly sampled where each element's pseudo-density is sampled from a truncated normal distribution with the means indicated above and a standard deviation of 0.2, truncated between 0 and 1. A value of $R = 10$ was used here.}
    \label{fig:correcting_std}
  \end{figure*}

\subsection{Mathematical analysis} \label{sec:analysis}

  In this section, an attempt will be made to mathematically explain the insensitivity of the estimators' correcting ratios to the design as shown above. While this section doesn't provide a rigorous proof of the phenomena observed, it does provide some mathematical insight into why it is happening and when it can be expected to happen in other problems.

  Let the diagonal estimator be:
  \begin{align}
  \frac{1}{N} \sum_{k = 1}^N \bm{D}_{\bm{v}_k} \bm{A} \bm{v}_k = 
  \frac{1}{N} \Bigl(\sum_{k = 1}^N \bm{D}_{\bm{v}_k} \bm{A} \bm{D}_{\bm{v}_k} \Bigr) \bm{1}
  \end{align}
  where $\bm{A} = \bm{F}^T \bm{K}^{-1} \bm{F}$. Let $a_{ij}$ be $(i,j)^{th}$ element of $\bm{A}$. Let $v_{ki}$ be the $i^{th}$ element of $\bm{v}_k$. The $(i,j)^{th}$ element of $\sum_{k = 1}^N \bm{D}_{\bm{v}_k} \bm{A} \bm{D}_{\bm{v}_k}$ is therefore $a_{ij} \sum_{k=1}^N v_{ki}v_{kj}$. Let $N^+_{ij}$ be the number of times $v_{ki}v_{kj}$ is 1 and $N^-_{ij}$ be the number of times $v_{ki}v_{kj}$ is -1. In the case of Hadamard basis, if $N$ is the smallest power of 2 larger than or equal to the number of loads $L$, then:
  \begin{align}
    N^+_{ij} = N^-_{ij} = N/2 & \quad \textit{if} \quad i \neq j \\
    N^+_{ij} = N, N^-_{ij} = 0 & \quad \textit{if} \quad i = j
  \end{align}
  This means that the diagonal estimate will be exact in that case. Bekas et al. \cite{Bekas2007} showed that Hadamard basis work well for banded matrices and for matrices where off-diagonal values are decaying rapidly away from the diagonal. However in the case of load compliances, neither of those conditions apply. Therefore, as shown in the experiment above, the accuracy of the estimated diagonal is quite bad as obvious from the standard deviation of the estimate. Let the $i^{th}$ diagonal element (or load compliance) be $C_i = a_{ii}$. The estimator of $a_{ii}$, $\hat{a}_{ii}$, can be written as:
  \begin{align}
    \hat{a}_{ii} & = \frac{1}{N} \sum_{j=1}^L a_{ij} \sum_{k=1}^N v_{ki}v_{kj} \\
    & = \frac{1}{N} \sum_{j=1}^L a_{ij} (N^+_{ij} - N^-_{ij}) \\
    & = a_{ii} + \sum_{j \neq i} a_{ij} \frac{N^+_{ij} - N^-_{ij}}{N}
  \end{align}
  The ratio of the estimated diagonal element to the actual diagonal element is:
  \begin{align}
    \frac{\hat{a}_{ii}}{a_{ii}} = 1 + \sum_{j \neq i} \frac{a_{ij}}{a_{ii}} \frac{N^+_{ij} - N^-_{ij}}{N}
  \end{align}
  This ratio depends on:
  \begin{enumerate}
    \item $\frac{a_{ij}}{a_{ii}}$ which depends on the design and the load scenarios, and
    \item $N^+_{ij}$ and $N^-_{ij}$ which depend on the Hadamard basis used. 
  \end{enumerate}
  If the same basis are used for all the designs during optimization, then $\frac{a_{ij}}{a_{ii}}$ is the only number that can vary.
  \begin{align}
    \frac{a_{ij}}{a_{ii}} = \frac{\bm{f}_i^T \bm{K}^{-1} \bm{f}_j}{\bm{f}_i^T \bm{K}^{-1} \bm{f}_i}
  \end{align}
  The partial derivative of $a_{ij} / a_{ii}$ with respect to the $e^{th}$ element's density $\rho_e$ is:
  \begin{align}
    \frac{\partial (a_{ij} / a_{ii})}{\partial \rho_e} & = \frac{\partial (\bm{f}_i^T \bm{K}^{-1} \bm{f}_j / \bm{f}_i^T \bm{K}^{-1} \bm{f}_i)}{\partial \rho_e} \\
    & = \frac{\frac{\partial \bm{f}_i^T \bm{K}^{-1} \bm{f}_j}{\partial \rho_e} \bm{f}_i^T \bm{K}^{-1} \bm{f}_i - \frac{\partial \bm{f}_i^T \bm{K}^{-1} \bm{f}_i}{\partial \rho_e} \bm{f}_i^T \bm{K}^{-1} \bm{f}_j}{(\bm{f}_i^T \bm{K}^{-1} \bm{f}_i)^2} \\
    & = \frac{-(\bm{u}_i^T \bm{K}_e \bm{u}_j) (\bm{u}_i^T \bm{K} \bm{u}_i) + (\bm{u}_i^T \bm{K}_e \bm{u}_i) (\bm{u}_i^T \bm{K} \bm{u}_j)}{(\bm{u}_i^T \bm{K} \bm{u}_i)^2} \\
  \end{align}

  \begin{lemma}
    \begin{multline}
      |\bm{u}_i^T \bm{K} \bm{u}_j| \leq \\ 
      \frac{1}{2} \max(|(\bm{u}_i + \bm{u}_j)^T \bm{K} (\bm{u}_i + \bm{u}_j)|, |\bm{u}_i^T \bm{K} \bm{u}_i + \bm{u}_j^T \bm{K} \bm{u}_j|)
    \end{multline}
    if $\bm{K}$ is positive or negative semi-definite and
    \begin{multline}
      |\bm{u}_i^T \bm{K} \bm{u}_j| \leq \\ 
      \frac{1}{2} \Big( |(\bm{u}_i + \bm{u}_j)^T \bm{K} (\bm{u}_i + \bm{u}_j)| + |\bm{u}_i^T \bm{K} \bm{u}_i| + |\bm{u}_j^T \bm{K} \bm{u}_j)| \Big)
    \end{multline}
    otherwise.
    \label{lemma:bound1}
  \end{lemma}
  \begin{proof}
    \begin{align}
      2 \bm{u}_i^T \bm{K} \bm{u}_j = (\bm{u}_i + \bm{u}_j)^T \bm{K} (\bm{u}_i + \bm{u}_j) - \bm{u}_i^T \bm{K} \bm{u}_i - \bm{u}_j^T \bm{K} \bm{u}_j
    \end{align}
    If $\bm{K}$ is indefinite:
    \begin{multline}
      2 |\bm{u}_i^T \bm{K} \bm{u}_j| \leq \\ |(\bm{u}_i + \bm{u}_j)^T \bm{K} (\bm{u}_i + \bm{u}_j)| + |\bm{u}_i^T \bm{K} \bm{u}_i| + |\bm{u}_j^T \bm{K} \bm{u}_j|
    \end{multline}
    If $\bm{K}$ is positive semi-definite, then $(\bm{u}_i + \bm{u}_j)^T \bm{K} (\bm{u}_i + \bm{u}_j)$, $\bm{u}_i^T \bm{K} \bm{u}_i$ and $\bm{u}_j^T \bm{K} \bm{u}_j$ are all non-negative. Therefore:
    \begin{multline}
      -(\bm{u}_i^T \bm{K} \bm{u}_i + \bm{u}_j^T \bm{K} \bm{u}_j) \leq \\ 2 \bm{u}_i^T \bm{K} \bm{u}_j \leq (\bm{u}_i + \bm{u}_j)^T \bm{K} (\bm{u}_i + \bm{u}_j)
    \end{multline}
    Similarly, if $\bm{K}$ is negative semi-definite, then \\ $(\bm{u}_i + \bm{u}_j)^T \bm{K} (\bm{u}_i + \bm{u}_j)$, $\bm{u}_i^T \bm{K} \bm{u}_i$ and $\bm{u}_j^T \bm{K} \bm{u}_j$ are all non-positive. Therefore:
    \begin{multline}
      (\bm{u}_i + \bm{u}_j)^T \bm{K} (\bm{u}_i + \bm{u}_j) \leq 2 \bm{u}_i^T \bm{K} \bm{u}_j \leq \\ -(\bm{u}_i^T \bm{K} \bm{u}_i + \bm{u}_j^T \bm{K} \bm{u}_j)
    \end{multline}
    It follows that:
    \begin{multline}
      2 |\bm{u}_i^T \bm{K} \bm{u}_j| \leq \\ \max(|(\bm{u}_i + \bm{u}_j)^T \bm{K} (\bm{u}_i + \bm{u}_j)|, |\bm{u}_i^T \bm{K} \bm{u}_i + \bm{u}_j^T \bm{K} \bm{u}_j|)
    \end{multline}
    This completes the proof.
  \end{proof}
  Using the above bound, it follows that if for all combinations of $i$ and $j$:
  \begin{align}
    \frac{\bm{u}_j^T \bm{K}_e \bm{u}_j}{\bm{u}_i^T \bm{K} \bm{u}_i} \leq \alpha_1 \\
    \frac{(\bm{u}_i + \bm{u}_j)^T \bm{K}_e (\bm{u}_i + \bm{u}_j)}{2 \bm{u}_i^T \bm{K} \bm{u}_i} \leq \beta_1 \\
    \frac{\bm{u}_j^T \bm{K} \bm{u}_j}{\bm{u}_i^T \bm{K} \bm{u}_i} \leq \alpha_2\\
    \frac{(\bm{u}_i + \bm{u}_j)^T \bm{K} (\bm{u}_i + \bm{u}_j)}{2 \bm{u}_i^T \bm{K} \bm{u}_i} \leq \beta_2
  \end{align}
  then
  \begin{align}
    \Bigl|\frac{\partial (a_{ij} / a_{ii})}{\partial \rho_e}\Bigr| & \leq \max(\alpha_1, \beta_1) + \alpha_1 \times \max(\alpha_2, \beta_2)
  \end{align}
  It is natural to expect $\alpha_1$ to be small since the element compliance due to any one load will likely be much smaller than the total compliance due to any other load. If the loading scenarios have widely varying magnitudes, $\alpha_1$ may be large in that case. However to remedy this, the loading scenarios can be clustered into groups by their norm and a separate estimator can be used for each group. Similarly, $\beta_1$ is likely to be small if all the forces have a close enough norm since the element compliance due to the superposition of 2 loads is likely to be much smaller than two times the total compliance due to any other load. If the norms of the loads are somewhat similar, $\alpha_2$ and $\beta_2$ can also be expected to be small constants greater than or equal to 1. This means that absolute value of the individual partial derivatives can be upper bounded by a small positive number. Interestingly, the sum of all the partial derivatives of the correcting ratio with respect to the individual element densities, $\sum_e \frac{\partial (a_{ij} / a_{ii})}{\partial \rho_e}$, is 0. This does not guarantee that the directional derivative in any direction will be small but it increases the chances of term cancellation. This is consistent with the observations.

  However, the correcting factor for the estimator $\hat{C}_i = \hat{a}_{ii}$ does not just depend on the individual $\frac{\partial (a_{ij} / a_{ii})}{\partial \rho_e}$ but rather it depends on the sum $\sum_{j \neq i} \frac{a_{ij}}{a_{ii}} \frac{N^+_{ij} - N^-_{ij}}{N}$. Three factors can make this sum small:
  \begin{enumerate}
    \item A good choice of probing vectors that make the distribution of $N^+_{ij} - N^-_{ij}$ for different $(i,j)$ pairs symmetric around 0 promoting term cancellation.
    \item Term cancellation due to the alternating signs of $a_{ij}$. For instance if the mean load vector is the $\bm{0}$ vector, the summation $\sum_{j \neq i} \frac{a_{ij}}{a_{ii}}$ is equal to -1 regardless of the number of loading scenarios.
    \item A small ratio of the number of loading scenarios to the number of elements. This is detailed below.
  \end{enumerate}
  For fixed loading scenarios, the values of $\alpha_1$ and $\beta_1$ decrease as the number of elements $E$ increases. This is because the ratio of an individual element's contribution to the total strain energy decreases as the element size decreases. Given that $-1 \leq \frac{N^+_{ij} - N^-_{ij}}{N} \leq 1$:
  \begin{align}
    \Biggl|\sum_{j \neq i} \frac{a_{ij}}{a_{ii}} \frac{N^+_{ij} - N^-_{ij}}{N} \Biggr| \leq (L - 1)(\beta_1 + \alpha_1 + \alpha_1 (\beta_2 + \alpha_2))
  \end{align}
  Therefore, if $L \ll E$ and the loads in $\bm{F}$ have close magnitudes, one can expect the correcting factor to be design insensitive especially near the end of the optimization when the design is not changing much. 

  The analysis above identified 3 strategies other than using more probing vectors that can help promote the insensitivity of the correcting factors to the design:
  \begin{enumerate}
    \item Clustering the loads by their magnitudes with a maximum number of loads per cluster $\ll E$,
    \item Centering the loads around $\bm{0}$. Let $\bm{\mu}_{\bm{f}}$ be the sample mean of the loading scenarios and let $\tilde{\bm{f}}_i = \bm{f}_i - \bm{\mu}_{\bm{f}}$. The $i^{th}$ load compliance $\bm{f}_i^T \bm{K}^{-1} \bm{f}_i$ would then be $\tilde{\bm{f}}_i^T \bm{K}^{-1} \tilde{\bm{f}}_i + 2 \tilde{\bm{f}}_i^T \bm{K}^{-1} \bm{\mu}_{\bm{f}} + \bm{\mu}_{\bm{f}}^T \bm{K}^{-1} \bm{\mu}_{\bm{f}}$. The terms $\tilde{\bm{f}}_i^T \bm{K}^{-1} \tilde{\bm{f}}_i$ can be obtained from the diagonal estimator of $\tilde{\bm{F}}^T \bm{K}^{-1} \tilde{\bm{F}}$ where the columns of $\tilde{\bm{F}}$ are the vectors $\tilde{\bm{f}}_i$. The remaining terms can be computed using a single linear system solve $\bm{K}^{-1} \bm{\mu}_{\bm{f}}$.
    \item Using a finer mesh, i.e. increasing $E$ thus decreasing $\alpha_1$ and $\beta_1$.
  \end{enumerate}

  Note that while the analysis above provides some mathematical insights into why the correcting ratios for the individual compliances may not be sensitive to the design, it is not a complete proof of the phenomena observed because only a single element's $\rho_e$ was assumed to be changing in the analysis. However, from the analysis above one can see that term cancellation is highly likely in practice. For instance, the sum of $\frac{a_{ij}}{a_{ii}}$ for all $j$ is equal to -1 if the mean load is $\bm{0}$ regardless of the number of loads, and the sum of $\frac{\partial a_{ij}/a_{ii}}{\partial \rho_e}$ for all $e$ is equal to 0. This term cancellation is the main reason behind the extreme insensitivity of the correcting ratio to the design observed in the experiments above even when all the elements' densities are changing in random directions by large amounts.

  Next it will be shown that under some conditions that the above insensitivity of the correcting ratio to any individual $\rho_e$ can be extended to a class of scalar-valued functions of the load compliances. This class of functions includes the mean, variance and standard deviation but not the augmented Lagrangian penalty. Let $\gamma_i$ be the correcting factor for the compliance $C_i$. The correcting factor of a scalar valued function $f$ of the load compliances can therefore be written as:
  \begin{align}
    \eta(\bm{\rho}) = \frac{f(\gamma_1(\bm{\rho}) \hat{C}_1(\bm{\rho}), \gamma_2(\bm{\rho}) \hat{C}_2(\bm{\rho}), \dots, \gamma_L(\bm{\rho}) \hat{C}_L(\bm{\rho}))}{f(\hat{C}_1(\bm{\rho}), \hat{C}_2(\bm{\rho}), \dots, \hat{C}_L(\bm{\rho}))}
  \end{align}
  Let $f_{\bm{\hat{C}}} = f(\hat{C}_1, \dots, \hat{C}_L)$ and $f_{\bm{C}} = f(\gamma_1 \hat{C}_1, \dots, \gamma_L \hat{C}_L)$. Furthermore, let $f_{\bm{\hat{C}}}^{(i)}$ be the partial derivative of $f$ with respect to its $i^{th}$ argument evaluated at $(\hat{C}_1, \hat{C}_2, \dots, \hat{C}_L)$ and let $f_{\bm{C}}^{(i)}$ be the partial derivative of $f$ with respect to its $i^{th}$ argument evaluated at $(\gamma_1 \hat{C}_1, \gamma_2 \hat{C}_2, \dots, \gamma_L \hat{C}_L)$.
  \begin{align}
    \frac{\partial \eta}{\partial \rho_e} & 
    = \sum_i \Biggl(\frac{\partial \eta}{\partial \gamma_i} * \frac{\partial \gamma_i}{\partial \rho_e} + 
    \frac{\partial \eta}{\partial \hat{C}_i} * \frac{\partial \hat{C}_i}{\partial \rho_e}\Biggr) \\
    & = \sum_i \Biggl( \frac{f_{\bm{C}}^{(i)} \hat{C}_i}{f_{\hat{\bm{C}}}} \frac{\partial \gamma_i}{\partial \rho_e} + \Biggl( \frac{f_{\bm{C}}^{(i)} \gamma_i}{f_{\bm{\hat{C}}}} - \frac{f_{\bm{\hat{C}}}^{(i)} f_{\bm{C}}}{f_{\hat{\bm{C}}}^2} \Biggr) \frac{\partial \hat{C}_i}{\partial \rho_e} \Biggr)
  \end{align}
  One can see that if the magnitudes of $f_{\hat{\bm{C}}}^{(i)}$ and $f_{\bm{C}}^{(i)}$ scale down as $L$ increases and if $f_{\bm{C}} / f_{\hat{\bm{C}}}^2$ is small that the partial derivative $\frac{\partial \eta}{\partial \rho_e}$ will also likely be small. For all $i$, let:
  \begin{align}
    |f_{\bm{C}}^{(i)}| \leq \frac{c_1}{L} \\
    |f_{\hat{\bm{C}}}^{(i)}| \leq \frac{c_1}{L} \\
    \Bigl| \frac{\partial \hat{C}_i}{\partial \rho_e} \Bigr| \leq c_2 \\
    \Bigl| \frac{\partial \gamma_i}{\partial \rho_e} \Bigr| \leq c_3 \\
    \Bigl| \frac{\hat{C}_i}{f_{\hat{\bm{C}}}} \Bigr| \leq c_4 \\
    \Bigl| \frac{\gamma_i}{f_{\hat{\bm{C}}}} \Bigr| \leq c_5 \\
    \Bigl| \frac{f_{\bm{C}}}{f_{\hat{\bm{C}}}^2} \Bigr| \leq c_6
  \end{align}
  Then one can set the bound:
  \begin{align}
    \Biggl| \frac{\partial \eta}{\partial \rho_e} \Biggr| \leq c_1 c_4 c_3 + c_1 c_5 c_2 + c_1 c_6 c_2
  \end{align}
  From the above bound, one can see that $c_3$, $c_5$ and $c_6$ must be small enough to guarantee a low upper bound on the absolute value of $\frac{\partial \eta}{\partial \rho_e}$. This means that:
  \begin{enumerate}
    \item The diagonal's correcting factors must not be sensitive to $\rho_e$ (i.e. $c_3$ is small). This has been established above under some conditions.
    \item The ratio of the diagonal correcting factors to the function estimator $f_{\hat{\bm{C}}}$ must be small in magnitude, i.e. ($c_5$ is small). This is true for the experiment above, where the diagonal correcting ratios at the full ground mesh ranged from -19.0 to 21.7 while the estimated compliance mean and standard deviation were 410.0 and 1329.7 respectively.
    \item The ratio $f_{\bm{C}} / f_{\hat{\bm{C}}}^2$ must be small in magnitude, (i.e. $c_6$ is small). This is also true for the experiment above at the full ground mesh where the ratios were $2.7e-3$ and $3.0e-4$ for the mean and standard deviation of the compliance respectively.
  \end{enumerate}

  To show that the above result applies to the mean, standard deviation and variance functions, it suffices to show that $|f_{\bm{C}}^{(i)}| \leq \frac{c_1}{L}$ for some constant $c_1$. If this is true for $f_{\bm{C}}^{(i)}$ then it is also true for $f_{\hat{\bm{C}}}^{(i)}$ since this is the same function evaluated at different points. The partial derivatives of the mean, standard deviation and variance of $(C_1, C_2, \dots, C_L)$ with respect to each $C_i$ are:
  \begin{align}
      \frac{\partial \mu_C}{\partial C_i} = \frac{1}{L}
  \end{align}
  \begin{align}
      \frac{\partial \sigma_C}{\partial C_i} = \Big( 1 - \frac{1}{L} \Big) \frac{C_i - \mu_C}{(L - 1) \times \sigma_C} \leq \frac{2(C_i - \mu_C)}{L \times \sigma_C}
  \end{align}
  \begin{align}
      \frac{\partial \sigma_C^2}{\partial C_i} = \Big( 1 - \frac{1}{L} \Big) \frac{2(C_i - \mu_C)}{(L - 1)} \leq \frac{4(C_i - \mu_C)}{L}
  \end{align}
  because $L - 1 \geq L / 2$ for $L > 1$. Let $l_{\mu_C}$ and $l_{\sigma_C}$ be lower bounds on $\mu_C$ and $\sigma_C$ for all the designs. The constant $c_1$ is therefore $1$ for $\mu_C$, $2(C_{max} - l_{\mu_C})/l_{\sigma_C}$ for $\sigma_C$ and $4(C_{max} - l_{\mu_C})$ for $\sigma_C^2$.

  Finally for the augmented Lagrangian function, it was not possible to establish the bound above. Even if the compliance constraints were scaled by $1/L$ allowing a bound of the form $c_1 / L$, $c_1$ would still scale up with the linear and quadratic penalties of the augmented Lagrangian function. The linear penalty is unbounded from above and the quadratic penalty grows exponentially during the optimization process. This means that no tight bound can be established. The experiments run were also consistent with this result where the diagonal estimation method was found to not work when solving a maximum compliance constrained problem using the augmented Lagrangian algorithm. A meaningless design was produced.

\section{Optimization}

  \subsection{Low rank loads}

    When minimizing the mean compliance only, the insensitivity of the correcting ratio to the design implies that one can minimize the mean compliance estimate instead of the exact one and get a reasonable design. This will be demonstrated in this section. In this section, a rank $R = 10$ is used and the trace estimation method is compared against the naive exact method where all the loading scenarios are enumerated. When minimizing the weighted sum of the mean and standard deviation of the compliance, a corrected estimator was used by calculating the correcting ratio of the mean and standard deviation estimators separately at the full ground mesh. Let the uncorrected mean compliance estimator be $\hat{\mu}_C$ and the uncorrected standard deviation estimator be $\hat{\sigma}_C$. The corrected estimator, $\hat{W}$, of the weighted sum of the mean and standard deviation used was:
    \begin{align} \label{eqn:corrected_estimator}
     \hat{W} = \frac{\mu_C(\bm{x}_0)}{\hat{\mu}_C(\bm{x}_0)} \hat{\mu}_C(\bm{x}) + 2 \frac{\sigma_C(\bm{x}_0)}{\hat{\sigma}_C(\bm{x}_0)} \hat{\sigma}_C(\bm{x})
    \end{align}
    Only Hadamard basis probing vectors were used in this section.

  \subsubsection{Mean compliance minimization}

    To demonstrate the effectiveness of the proposed approaches, the cantilever beam problem described in section \ref{sec:exp_impl} was solved using the proposed exact and approximate methods. Table \ref{tab:mean_stats} shows the statistics of the final optimal solutions obtained by minimizing the mean compliance subject to the volume fraction constraint using exact and trace estimation methods to evaluate the mean compliance. 10 Hadamard basis probing vectors were used in the trace estimator. The optimal topologies are shown in Figure \ref{fig:mean}.

    While the designs obtained were different, both algorithms converged to reasonable designs in similar amounts of time. The convergence time shows that the convergence behavior was not affected by the use of an estimator in place of the original objective. However, the design produced by the trace estimation method was significantly worse than the exact method's which is to be expected since an approximate objective was minimized. Finally, note that the correcting ratio of the mean compliance estimator at the final design is 1.276 which is very close to the values shown in Figure \ref{fig:correcting_mean}.

    \begin{figure*}
      \begin{subfigure}[t]{0.45\textwidth}
        \centering
        \includegraphics[width=1\textwidth]{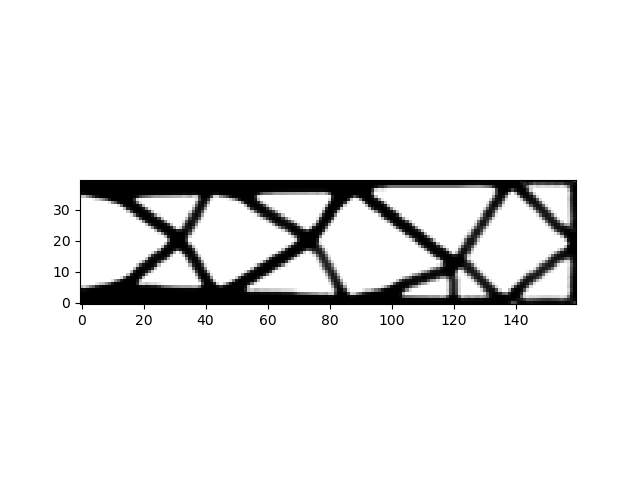}
        \caption{Exact method.}
      \end{subfigure} \hfill
      \begin{subfigure}[t]{0.45\textwidth}
        \centering
        \includegraphics[width=1\textwidth]{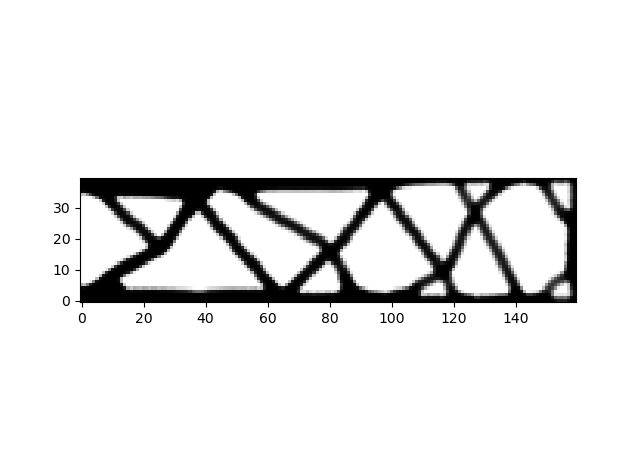}
        \caption{Trace estimation method with 10 Hadamard basis probing vectors.}
      \end{subfigure}
      \caption{Optimal topologies of the mean compliance minimization problem using continuation SIMP.}
      \label{fig:mean}
    \end{figure*}

    \begin{table*}
     \centering
     \caption{Summary statistics of the load compliances of the optimal solution of the mean compliance minimization problem using the exact and trace estimation methods to evaluate the mean compliance. 10 Hadamard basis probing vectors were used in the trace estimator.}
     \begin{tabular}{|c|c|c|}
      \hline
      \multirow{2}{5em}{Compliance Stat} & \multicolumn{2}{c|}{Value} \\\cline{2-3}
      & Exact & Trace estimation \\
      \hline
      \hline
      $\mu_C$ (Nmm) & 9397.4 & 7534.1 (uncorrected approx) / 9563.4 (exact) \\
      \hline
      $\sigma_C$ (Nmm) & 9689.0 & 9698.8 \\
      \hline
      $C_{max}$ (Nmm) & 125440.7 & 124021.8 \\
      \hline
      $C_{min}$ (Nmm) & 468.7 & 380.8 \\
      \hline
      $V$ & 0.400 & 0.400 \\
      \hline
      $Time$ (s) & 25505.4 & 375.5 \\
      \hline
     \end{tabular}
     \label{tab:mean_stats}
    \end{table*}

  \subsubsection{Mean-std compliance minimization}

    Similarly, Table \ref{tab:mean_std_stats} shows the statistics of the final solutions of the mean-std minimization problem solved using the exact and the corrected diagonal estimator method with 10 Hadamard basis probing vectors. The optimal topologies are shown in Figure \ref{fig:mean_std}. Both algorithms converged to reasonable, feasible designs. Additionally, as expected the exact and approximate mean-std minimization algorithms converged to solutions with lower compliance standard deviations but higher means compared to the exact and approximate mean minimization algorithms. It should be noted that the approximation error and non-convexity of the problem can sometimes lead this expectation to be unmet with the approximate approaches. The results indicate that the approximate method is able to converge in a fraction of the time it takes the exact method to converge because evaluating the function and its gradient using diagonal estimation requires $2N = 20$ linear system solves while the naive exact method requires 1000. This problem uses a low rank $\bm{F}$. The results of using the approximate methods proposed to solve a problem with a load matrix $\bm{F}$ of rank 100 are shown in the next section.

    \begin{figure*}
      \begin{subfigure}[t]{0.45\textwidth}
        \centering
        \includegraphics[width=1\textwidth]{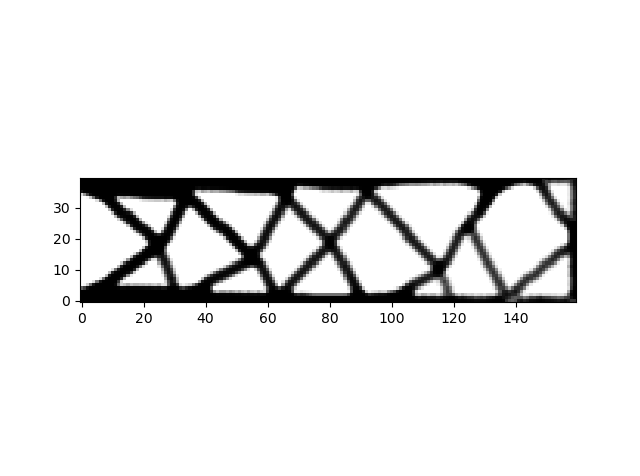}
        \caption{Exact method.}
      \end{subfigure} \hfill
      \begin{subfigure}[t]{0.45\textwidth}
        \centering
        \includegraphics[width=1\textwidth]{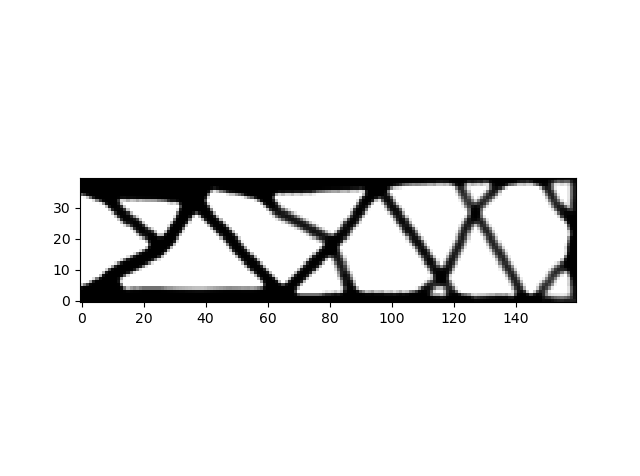}
        \caption{Corrected diagonal estimation method with 10 Hadamard basis probing vectors using the estimator in Eq. \ref{eqn:corrected_estimator}.}
      \end{subfigure}
      \caption{Optimal topologies of the mean-std compliance minimization problem using continuation SIMP.}
      \label{fig:mean_std}
    \end{figure*}

    \begin{table*}[!htbp]
     \centering
     \caption{Summary statistics of the load compliances of the optimal solution of the mean-std compliance minimization problem using exact and the corrected diagonal estimation method with 10 Hadamard basis probing vectors to evaluate the mean-std compliance.}
     \begin{tabular}{|c|c|c|}
      \hline
      \multirow{2}{5em}{Compliance Stat} & \multicolumn{2}{c|}{Value} \\\cline{2-3}
      & Exact & Diagonal estimation \\
      \hline
      \hline
      $\mu_C (Nmm) $ & 9871.6 & 9906.6 (corrected approx) / 9809.8 (exact) \\
      \hline
      $\sigma_C (Nmm) $ & 9264.0 & 9225.4 (corrected approx) / 9348.6 (exact) \\
      \hline
      $\mu_C + 2.0 \sigma_C (Nmm) $ & 28407.6 & 28357.4 (corrected approx) / 28513.4 (exact) \\
      \hline
      $C_{max}$ (Nmm) & 117956.4 & 118853.9 \\
      \hline
      $C_{min}$ (Nmm) & 530.2 & 451.9 \\
      \hline
      $V$ & 0.400 & 0.400 \\
      \hline
      Time (s) & 2821.1 & 540.0 \\
      \hline
     \end{tabular}
     \label{tab:mean_std_stats}
    \end{table*}

  \subsection{High rank loads}

    In this section, the 2D problems solved above will be solved using a load scenarios matrix $\bm{F}$ of rank $R = 100$ instead of 10. Additionally, the SVD-based method proposed by \cite{tarek2021robust} will be used instead of the naive approach used above. This will highlight the disadvantage of the SVD-based method when using a high rank $\bm{F}$.
  
    The results are shown to be consistent with the low rank $\bm{F}$ where the corrected estimator's accuracy is significantly improved by a single correction at the beginning of the optimization. The histograms in Figures \ref{fig:correcting_mean_high_rank} and \ref{fig:correcting_std_high_rank} also suggest that the correcting ratio is insensitive to the design. Figure \ref{fig:exact_approx_high_rank} shows that the Hadamard probing vectors do not always give a better estimator than the Rademacher-distributed one for the mean but it is consistently better for the standard deviation. Figures \ref{fig:mean_high_rank} and \ref{fig:mean_std_high_rank} and tables \ref{tab:mean_stats_high_rank} and \ref{tab:mean_std_stats_high_rank} show the optimal topologies and results obtained using the exact and approximate methods. The results are consistent with the expectations.
    \begin{figure*}
      \begin{subfigure}[t]{0.45\textwidth}
        \centering
        \includegraphics[width=1\textwidth]{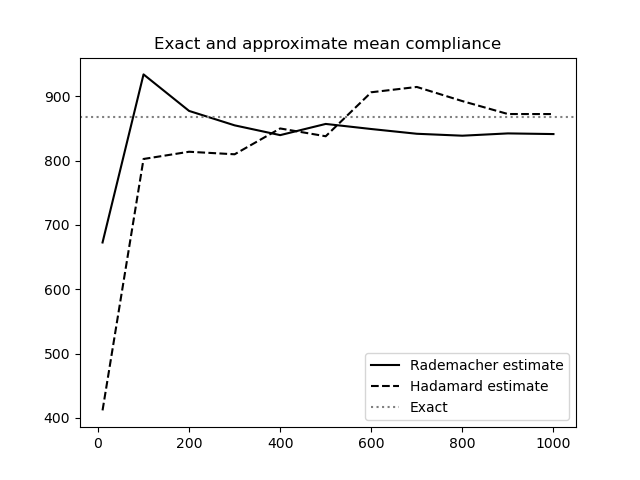}
        \caption{Mean compliance estimate using different numbers of probing vectors in the trace estimation method.}
        \label{fig:exact_approx_mean_high_rank}
      \end{subfigure} \hfill
      \begin{subfigure}[t]{0.45\textwidth}
        \centering
        \includegraphics[width=1\textwidth]{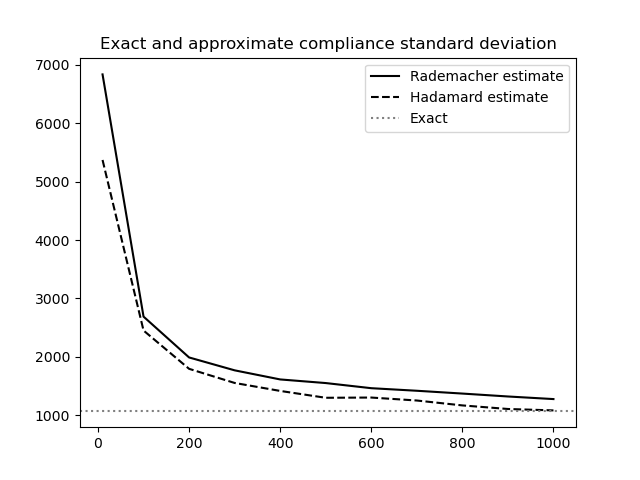}
        \caption{Compliance standard deviation estimate using different numbers of probing vectors in the diagonal estimation method.}
        \label{fig:exact_approx_std_high_rank}
      \end{subfigure}
      \caption{Accuracy profile of the trace and diagonal estimation methods for estimating the mean compliance and its standard deviation using 10, 100, 200, 300, 400, 500, 600, 700, 800, 900 and 1000 probing vectors for the high rank $\bm{F}$ case.}
      \label{fig:exact_approx_high_rank}
    \end{figure*}

    \begin{figure*}
      \begin{subfigure}[t]{0.3\textwidth}
        \centering
        \includegraphics[width=1\textwidth]{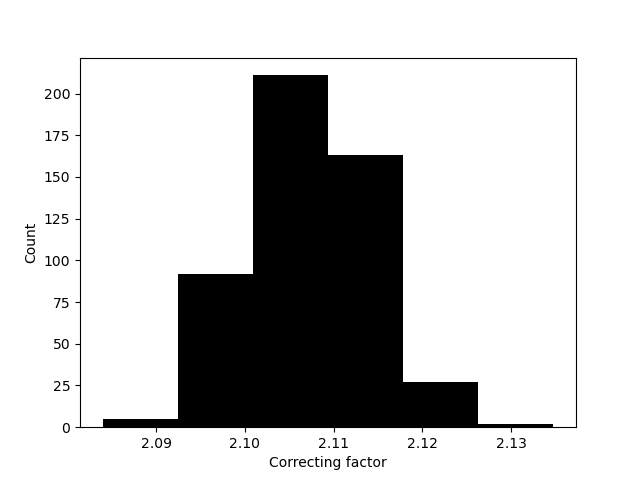}
        \caption{Mean = 0.1}
      \end{subfigure}
      \begin{subfigure}[t]{0.3\textwidth}
        \centering
        \includegraphics[width=1\textwidth]{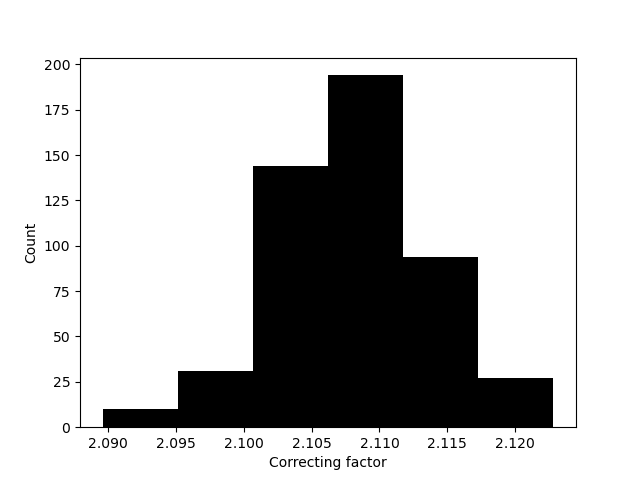}
        \caption{Mean = 0.3}
      \end{subfigure}
      \begin{subfigure}[t]{0.3\textwidth}
        \centering
        \includegraphics[width=1\textwidth]{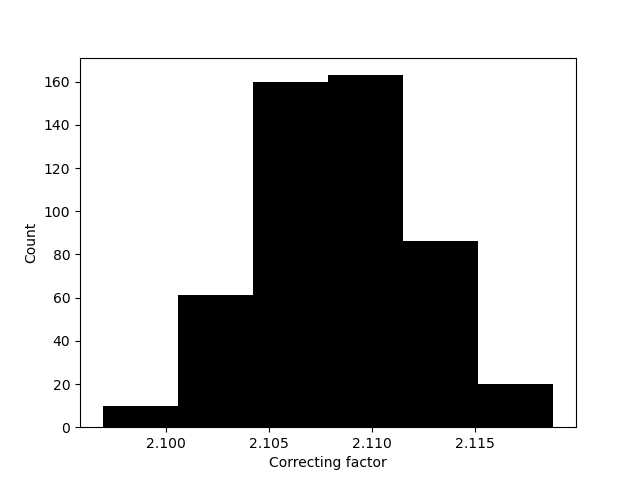}
        \caption{Mean = 0.5}
      \end{subfigure}
      \begin{subfigure}[t]{0.3\textwidth}
        \centering
        \includegraphics[width=1\textwidth]{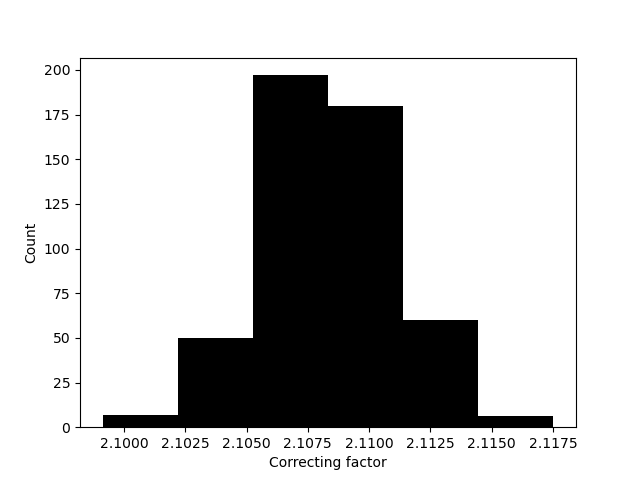}
        \caption{Mean = 0.7}
      \end{subfigure}
      \begin{subfigure}[t]{0.3\textwidth}
        \centering
        \includegraphics[width=1\textwidth]{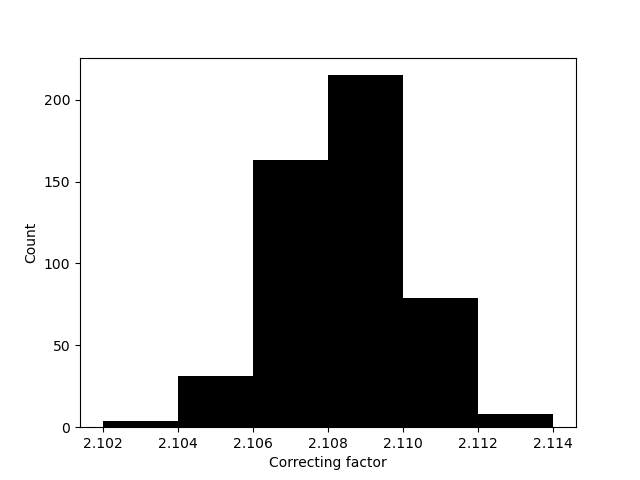}
        \caption{Mean = 0.9}
      \end{subfigure}
      \caption{Histograms of the ratio between the exact mean compliance and the trace estimate using 10 Hadamard basis probing vectors for the high rank $\bm{F}$. In each figure, 500 designs were randomly sampled where each element's pseudo-density is sampled from a truncated normal distribution with the means indicated above and a standard deviation of 0.2, truncated between 0 and 1.}
      \label{fig:correcting_mean_high_rank}
    \end{figure*}
    \begin{figure*}
      \begin{subfigure}[t]{0.3\textwidth}
        \centering
        \includegraphics[width=1\textwidth]{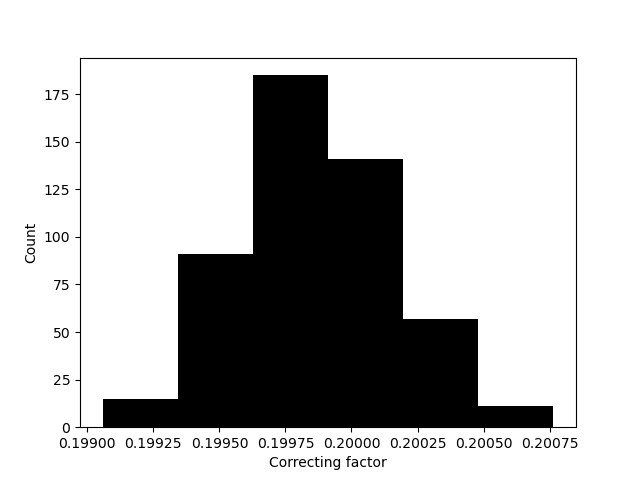}
        \caption{Mean = 0.1}
      \end{subfigure}
      \begin{subfigure}[t]{0.3\textwidth}
        \centering
        \includegraphics[width=1\textwidth]{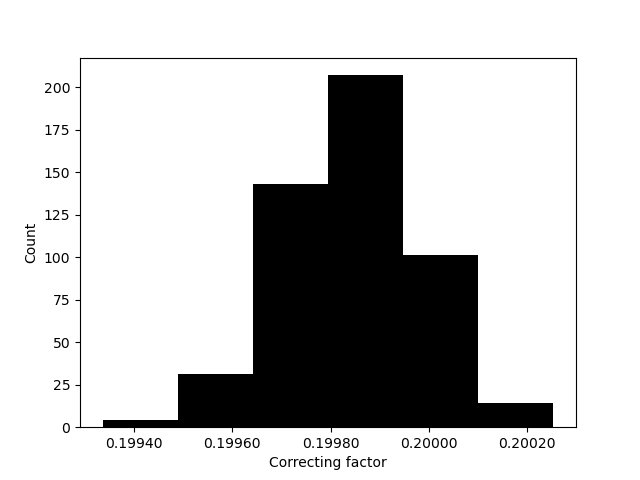}
        \caption{Mean = 0.3}
      \end{subfigure}
      \begin{subfigure}[t]{0.3\textwidth}
        \centering
        \includegraphics[width=1\textwidth]{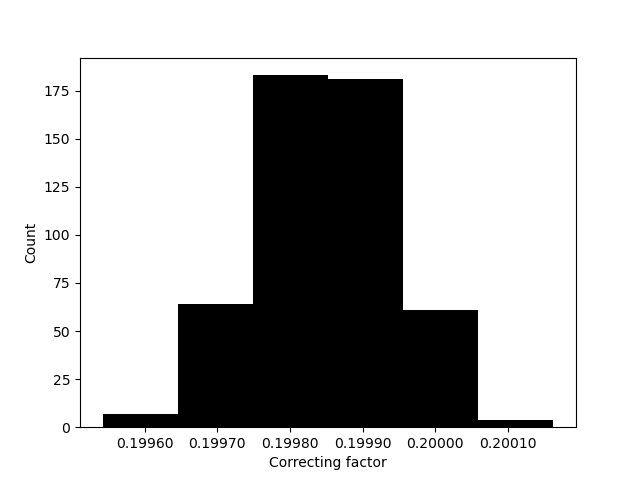}
        \caption{Mean = 0.5}
      \end{subfigure}
      \begin{subfigure}[t]{0.3\textwidth}
        \centering
        \includegraphics[width=1\textwidth]{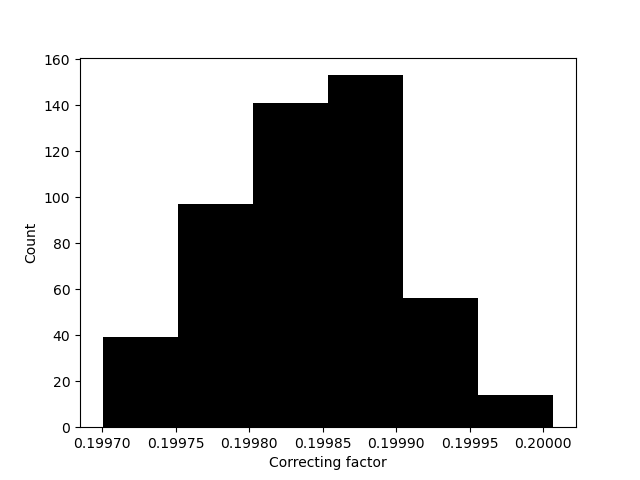}
        \caption{Mean = 0.7}
      \end{subfigure}
      \begin{subfigure}[t]{0.3\textwidth}
        \centering
        \includegraphics[width=1\textwidth]{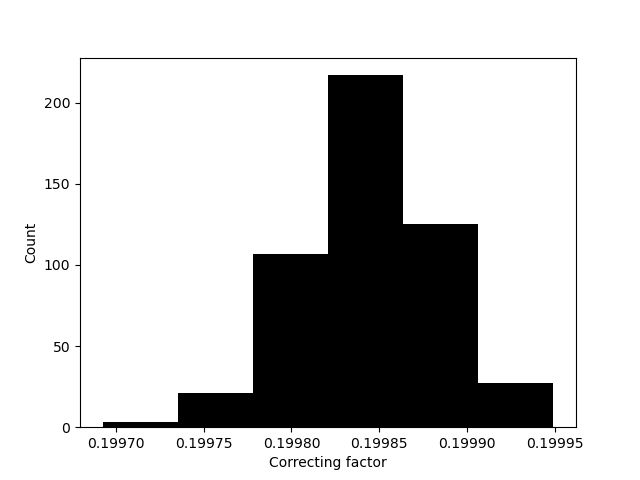}
        \caption{Mean = 0.9}
      \end{subfigure}
      \caption{Histograms of the ratio between the exact compliance standard deviation and the estimate using 10 Hadamard basis probing vectors for the high rank $\bm{F}$ case. In each figure, 500 designs were randomly sampled where each element's pseudo-density is sampled from a truncated normal distribution with the means indicated above and a standard deviation of 0.2, truncated between 0 and 1.}
      \label{fig:correcting_std_high_rank}
    \end{figure*}

    \begin{figure*}
      \begin{subfigure}[t]{0.45\textwidth}
        \centering
        \includegraphics[width=1\textwidth]{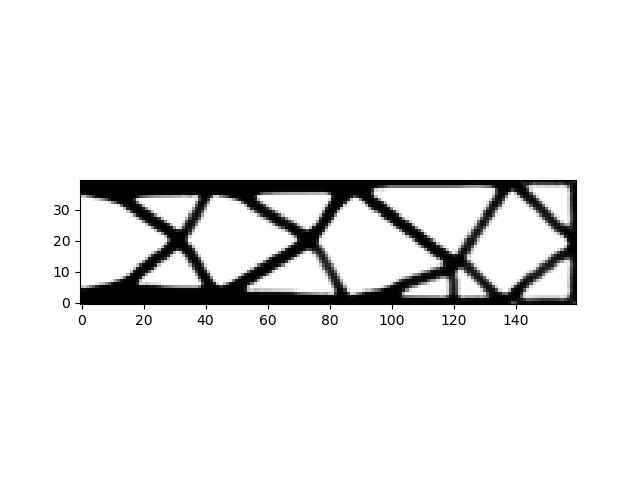}
        \caption{Exact method.}
      \end{subfigure} \hfill
      \begin{subfigure}[t]{0.45\textwidth}
        \centering
        \includegraphics[width=1\textwidth]{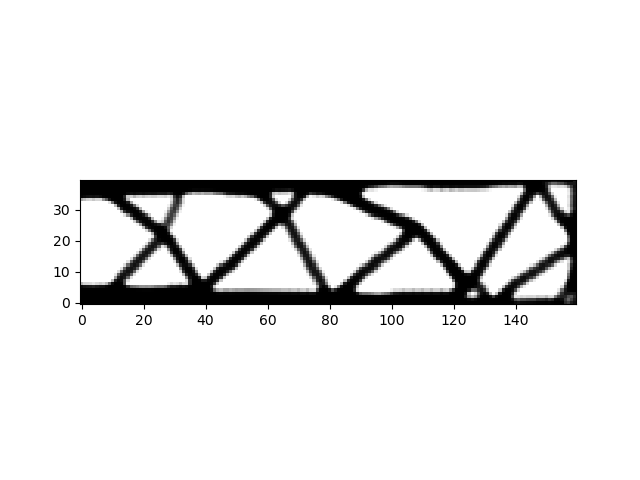}
        \caption{Trace estimation method with 10 Hadamard basis probing vectors.}
      \end{subfigure}
      \caption{Optimal topologies of the mean compliance minimization problem with a high rank $\bm{F}$ using continuation SIMP.}
      \label{fig:mean_high_rank}
    \end{figure*}

    \begin{table*}[!htbp]
     \centering
     \caption{Summary statistics of the load compliances of the optimal solution of the mean compliance minimization problem with a high rank $\bm{F}$ using the exact and trace estimation methods to evaluate the mean compliance. 10 Hadamard basis probing vectors were used in the trace estimator.}
     \begin{tabular}{|c|c|c|}
      \hline
      \multirow{2}{5em}{Compliance Stat} & \multicolumn{2}{c|}{Value} \\\cline{2-3}
      & Exact & Trace estimation \\
      \hline
      \hline
      $\mu_C$ (Nmm) & 2084.3 & 1094.1 (uncorrected approx) / 2226.4 (exact) \\
      \hline
      $\sigma_C$ (Nmm) & 2226.6 & 2412.0 \\
      \hline
      $C_{max}$ (Nmm) & 15971.9 & 16793.3 \\
      \hline
      $C_{min}$ (Nmm) & 171.0 & 152.0 \\
      \hline
      $V$ & 0.400 & 0.400 \\
      \hline
      $Time$ (s) & 2123.0 & 485.3 \\
      \hline
     \end{tabular}
     \label{tab:mean_stats_high_rank}
    \end{table*}

    \begin{figure*}
      \begin{subfigure}[t]{0.45\textwidth}
        \centering
        \includegraphics[width=1\textwidth]{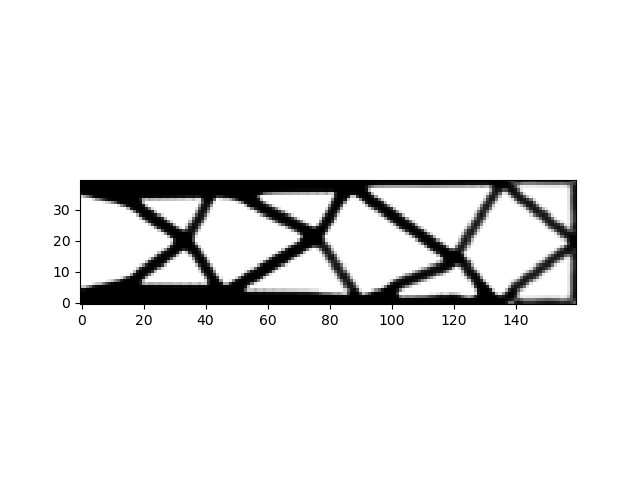}
        \caption{Exact method.}
      \end{subfigure} \hfill
      \begin{subfigure}[t]{0.45\textwidth}
        \centering
        \includegraphics[width=1\textwidth]{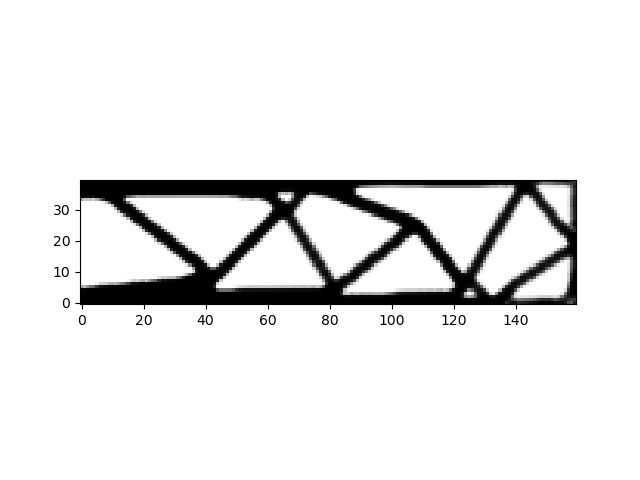}
        \caption{Corrected diagonal estimation method with 10 Hadamard basis probing vectors using the estimator in Eq. \ref{eqn:corrected_estimator}.}
      \end{subfigure}
      \caption{Optimal topologies of the mean-std compliance minimization problem with high rank $\bm{F}$ using continuation SIMP.}
      \label{fig:mean_std_high_rank}
    \end{figure*}

    \begin{table*}[!htbp]
     \centering
     \caption{Summary statistics of the load compliances of the optimal solution of the mean-std compliance minimization problem with a high rank $\bm{F}$ using exact and the corrected diagonal estimation method with 10 Hadamard basis probing vectors to evaluate the mean-std compliance.}
     \begin{tabular}{|c|c|c|}
      \hline
      \multirow{2}{5em}{Compliance Stat} & \multicolumn{2}{c|}{Value} \\\cline{2-3}
      & Exact & Diagonal estimation \\
      \hline
      \hline
      $\mu_C$ (Nmm) & 2151.9 & 2336.8 (corrected approx) / 2147.5 (exact) \\
      \hline
      $\sigma_C$ (Nmm) & 2149.0 & 2195.7 (corrected approx) / 2234.2 (exact) \\
      \hline
      $\mu_C + 2.0 \sigma_C$ (Nmm) & 6450.6 & 6728.3 (corrected approx) / 6616.7 (exact) \\
      \hline
      $C_{max}$ (Nmm) & 15558.2 & 15559.4 \\
      \hline
      $C_{min}$ (Nmm) & 187.4 & 161.6 \\
      \hline
      $V$ & 0.400 & 0.400 \\
      \hline
      Time (s) & 6435.4 & 650.1 \\
      \hline
     \end{tabular}
     \label{tab:mean_std_stats_high_rank}
    \end{table*}

    As shown in tables \ref{tab:mean_stats_high_rank} and \ref{tab:mean_std_stats_high_rank}, the SVD-based methods are slower than the approximation schemes proposed when the rank of the loads is high. This is because the number of non-zero singular values will be 100 which is 10x the number of probing vectors used. In the mean compliance minimization, a 10 speedup is achieved which is consistent with the expectation. In the mean-std compliance minimization, the diagonal estimation method requires 20 linear system solves so only a factor of 5 speedup is achieved with the approximate method compared to the SVD-based method.

  \subsection{3D cantilever beam problem}

    A 3D version of the 2D cantilever beam test problem used above was also solved using the methods proposed in this paper. The problem settings are described and the results are shown below.

    A 60 mm x 20 mm x 20 mm 3D cantilever beam was used with hexahedral elements of cubic shape and side length of 1 mm. The loads $\bm{F}_1$, $\bm{F}_2$ and $\bm{F}_3$ were positioned at (60, 10, 10), (30, 20, 10) and (40, 0, 10) where the coordinates represent the length, height and depth respectively. A value of $R = 10$ was used. The remaining loads and multipliers were sampled from the same distributions as the 2D problem. A density filter radius of 3 mm was also used for the 3D problem. The same volume constrained mean compliance minimization and volume constrained mean-std compliance minimization problems were solved.

    \subsubsection{Mean compliance minimization}

      \begin{figure*}
        \begin{subfigure}[t]{0.45\textwidth}
          \centering
          \includegraphics[width=1\textwidth]{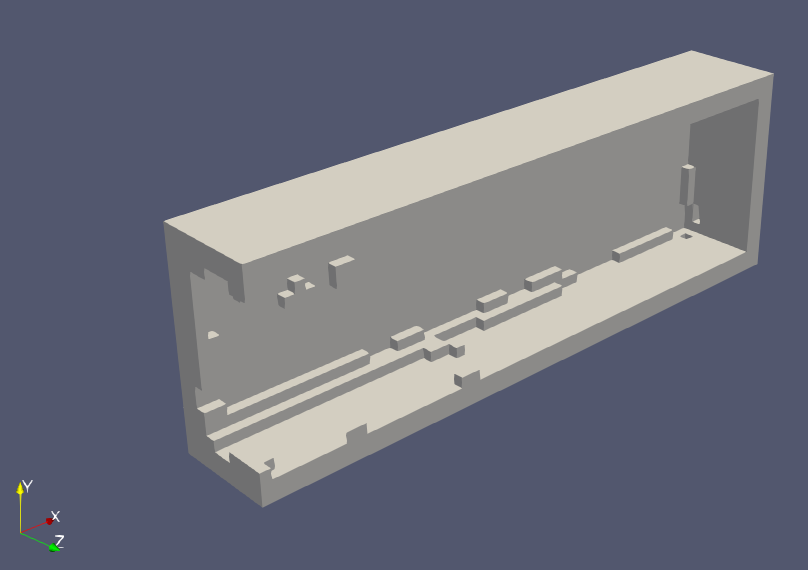}
          \caption{Left half}
        \end{subfigure} \hfill
        \begin{subfigure}[t]{0.45\textwidth}
          \centering
          \includegraphics[width=1\textwidth]{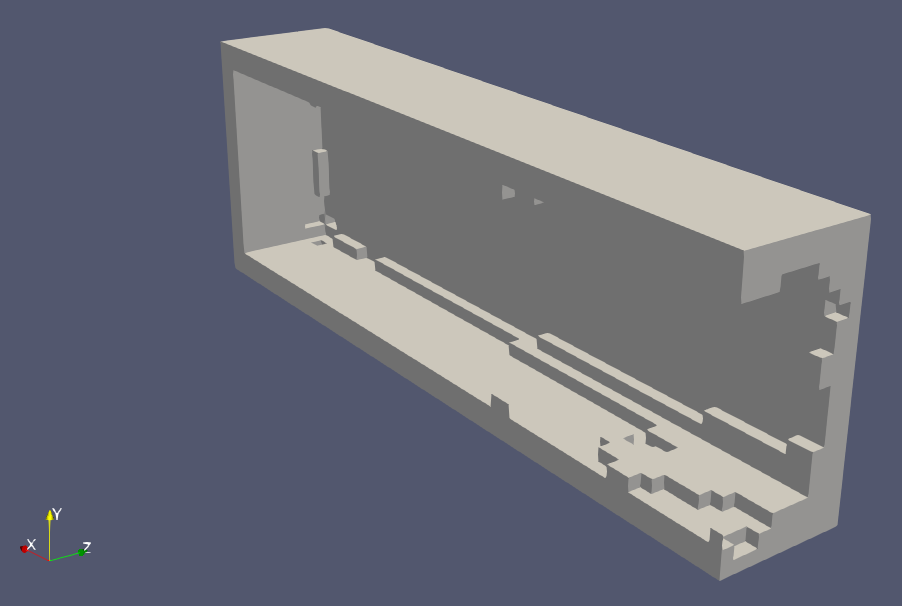}
          \caption{Right half}
        \end{subfigure}
        \caption{Cut views of the optimal topologies of the 3D mean compliance minimization problem using exact method.}
        \label{fig:exact_mean_3d}
      \end{figure*}

      \begin{figure*}
        \begin{subfigure}[t]{0.45\textwidth}
          \centering
          \includegraphics[width=1\textwidth]{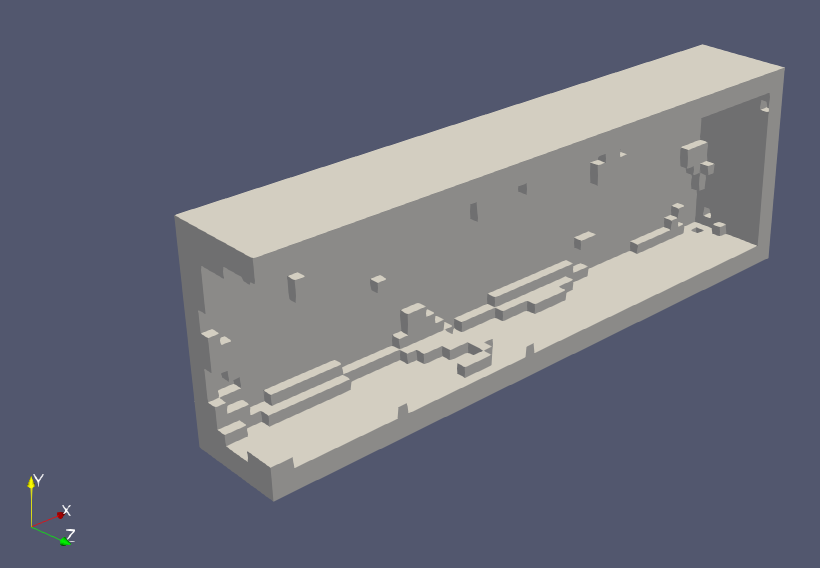}
          \caption{Left half}
        \end{subfigure} \hfill
        \begin{subfigure}[t]{0.45\textwidth}
          \centering
          \includegraphics[width=1\textwidth]{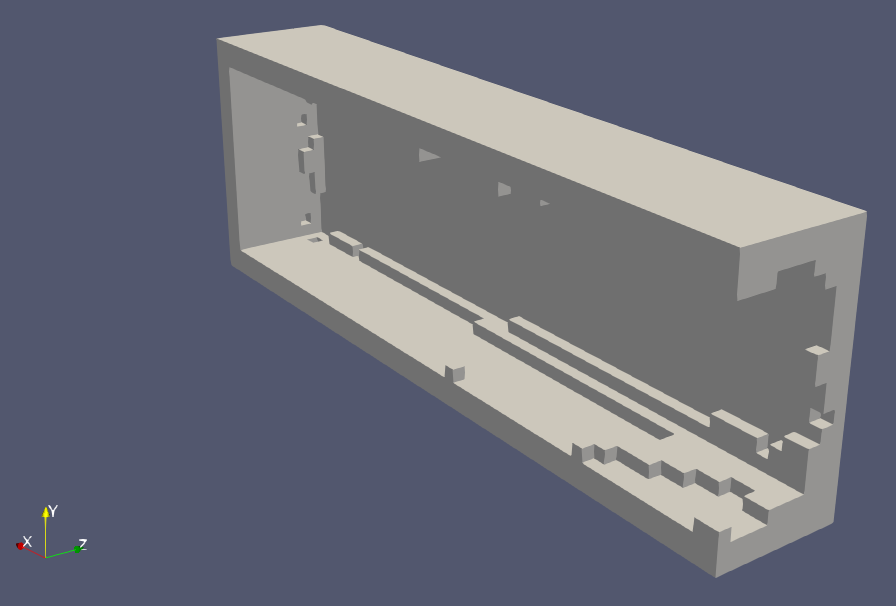}
          \caption{Right half}
        \end{subfigure}
        \caption{Cut views of the optimal topologies of the 3D mean compliance minimization problem using the trace estimation method.}
        \label{fig:trace_mean_3d}
      \end{figure*}

      \begin{table*}[!htbp]
       \centering
       \caption{Summary statistics of the load compliances of the optimal solution of the 3D mean compliance minimization problem using the exact and trace estimation methods to evaluate the mean compliance. 10 Hadamard basis probing vectors were used in the trace estimator.}
       \begin{tabular}{|c|c|c|}
        \hline
        \multirow{2}{5em}{Compliance Stat} & \multicolumn{2}{c|}{Value} \\\cline{2-3}
        & Exact & Trace estimation \\
        \hline
        \hline
        $\mu_C$ (Nmm) & 22072.1 & 24710.5 (uncorrected approx) / 22264.3 (exact) \\
        \hline
        $\sigma_C$ (Nmm) & 16628.7 & 17055.2 \\
        \hline
        $C_{max}$ (Nmm) & 184055.0 & 190599.9 \\
        \hline
        $C_{min}$ (Nmm) & 1785.8 & 1790.9 \\
        \hline
        $V$ & 0.400 & 0.400 \\
        \hline
        $Time$ (s) & 167321.2 & 6595.2 \\
        \hline
       \end{tabular}
       \label{tab:mean_stats_3d}
      \end{table*}

      The 3D cantilever beam problem described above was solved using the proposed approximate methods with the objective of minimizing the mean compliance subject to a volume fraction constraint with a limit of 0.4. Table \ref{tab:mean_stats_3d} shows the statistics of the final optimal solutions obtained by minimizing the mean compliance subject to the volume fraction constraint using the naive exact approach and the trace estimation method to evaluate the mean compliance. 10 Hadamard basis probing vectors were used in the trace estimator. The optimal topologies are shown in Figures \ref{fig:exact_mean_3d} and \ref{fig:trace_mean_3d}. Similar results to the 2D case can be observed where the designs obtained are different but somewhat reasonable. The proposed method converged in a small fraction of the time that the naive method took to converge. However, the design produced by the trace estimation method was worse than the exact method's which is to be expected since an approximate objective was minimized. Finally, note that the corrected estimate is close to the exact value.

    \subsubsection{Mean-std compliance minimization}

      \begin{figure*}
        \begin{subfigure}[t]{0.45\textwidth}
          \centering
          \includegraphics[width=1\textwidth]{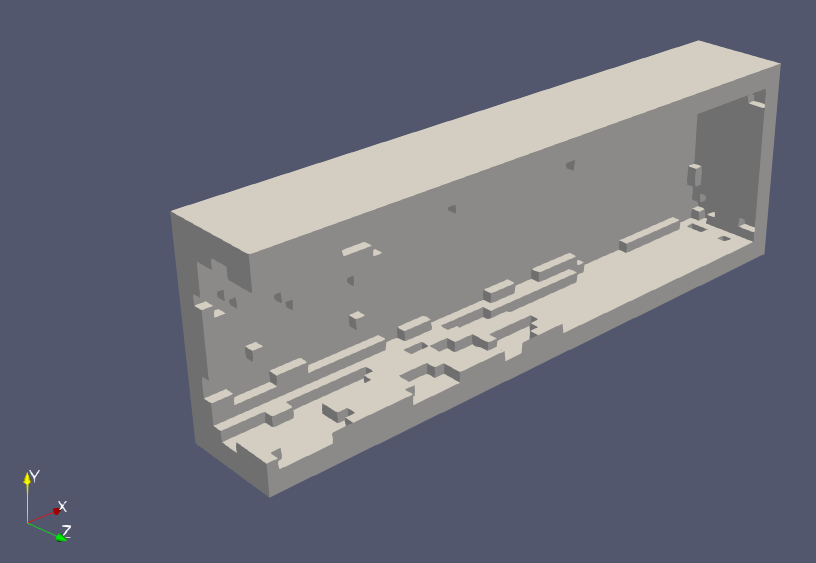}
          \caption{Left half}
        \end{subfigure} \hfill
        \begin{subfigure}[t]{0.45\textwidth}
          \centering
          \includegraphics[width=1\textwidth]{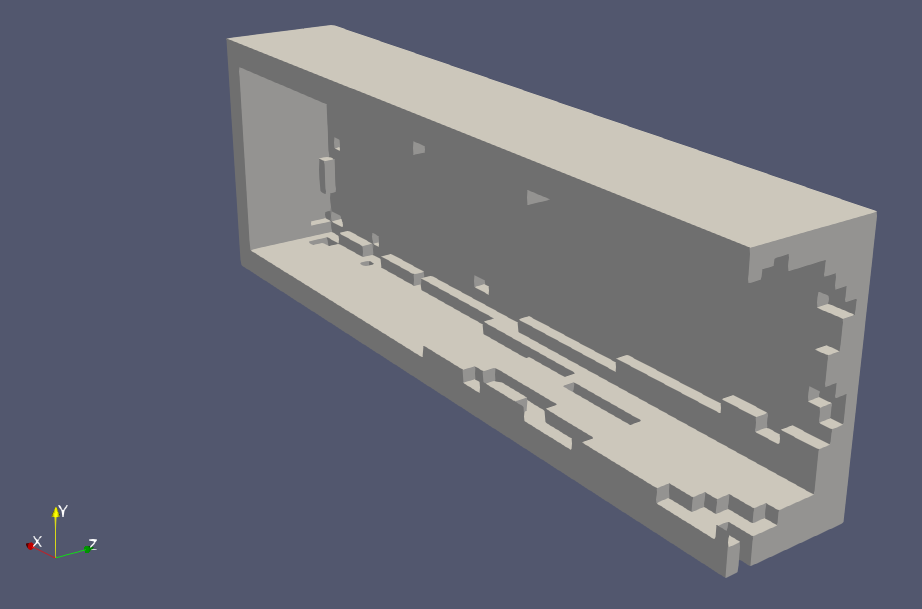}
          \caption{Right half}
        \end{subfigure}
        \caption{Cut views of the optimal topologies of the 3D mean-std compliance minimization problem using the exact method.}
        \label{fig:exact_mean_std_3d}
      \end{figure*}

      \begin{figure*}
        \begin{subfigure}[t]{0.45\textwidth}
          \centering
          \includegraphics[width=1\textwidth]{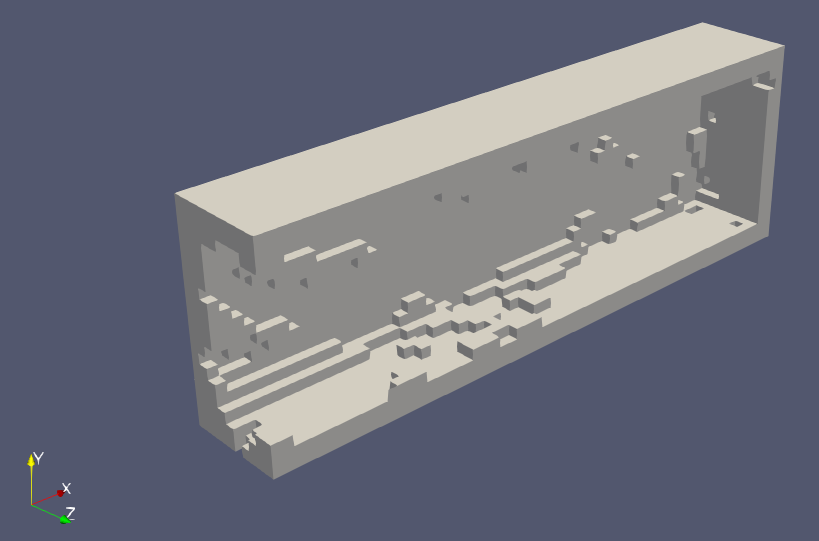}
          \caption{Left half}
        \end{subfigure} \hfill
        \begin{subfigure}[t]{0.45\textwidth}
          \centering
          \includegraphics[width=1\textwidth]{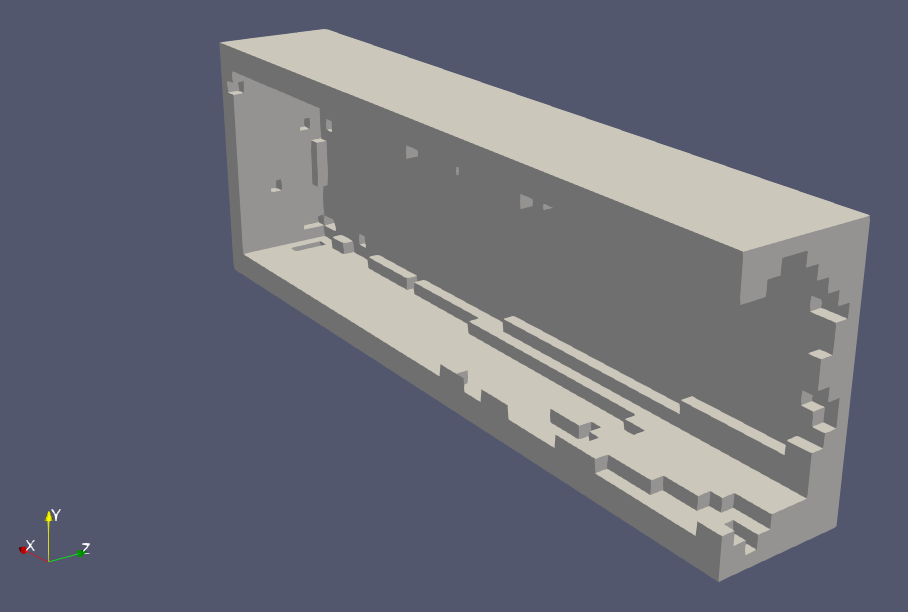}
          \caption{Right half}
        \end{subfigure}
        \caption{Cut views of the optimal topologies of the 3D mean-std compliance minimization problem using the corrected diagonal estimation method.}
        \label{fig:diagonal_mean_std_3d}
      \end{figure*}

      \begin{table*}[!htbp]
       \centering
       \caption{Summary statistics of the load compliances of the optimal solution of the mean-std compliance minimization problem using exact and the corrected diagonal estimation method with 10 Hadamard basis probing vectors to evaluate the mean-std compliance.}
       \begin{tabular}{|c|c|c|}
        \hline
        \multirow{2}{5em}{Compliance Stat} & \multicolumn{2}{c|}{Value} \\\cline{2-3}
        & Exact & Diagonal estimation \\
        \hline
        \hline
        $\mu_C$ (Nmm) & 22216.7 & 2240.8 (corrected approx) / 22145.9 (exact) \\
        \hline
        $\sigma_C$ (Nmm) & 16220.2 & 16501.4 (corrected approx) / 16510.0 (exact) \\
        \hline
        $\mu_C + 2.0 \sigma_C$ (Nmm) & 54848.8 & 55423.7 (corrected approx) / 55366.0 (exact) \\
        \hline
        $C_{max}$ (Nmm) & 176153.2 & 182209.8 \\
        \hline
        $C_{min}$ (Nmm) & 1872.0 & 1850.4 \\
        \hline
        $V$ & 0.400 & 0.400 \\
        \hline
        Time (s) & 39935.4 & 10949.3 \\
        \hline
       \end{tabular}
       \label{tab:mean_std_stats_3d}
      \end{table*}

      Similarly, Table \ref{tab:mean_std_stats_3d} shows the statistics of the final solutions of the 3D mean-std minimization problem solved using the naive exact approach and the corrected diagonal estimator method with 10 Hadamard basis probing vectors. The optimal topologies are shown in Figures \ref{fig:exact_mean_std_3d} and \ref{fig:diagonal_mean_std_3d}. Both algorithms converged to reasonable and feasible designs. Additionally, as expected the exact mean-std minimization converged to a solution with a lower standard deviation but a higher mean compliance compared to the exact mean minimization. However, due to the approximation error and non-convexity of the problems, the exact and approximate mean-std algorithms converged to solutions with a lower mean and std compared to the approximate mean algorithm. Finally as expected, the exact method took significantly longer to converge than the diagonal estimation method.

\section{Conclusion} \label{sec:conclusion}

  In this paper, two approximate methods were proposed to handle load uncertainty in compliance topology optimization problems where the uncertainty is described in the form of a set of finitely many loading scenarios. By re-formulating the function as a trace or diagonal estimation problem, significant performance improvements were achieved over the exact methods. Such improvement was demonstrated via complexity analysis and computational experiments. The methods proposed were shown to work well in practice while having a different time complexity profile.

\section{Acknowledgments}

  This research did not receive any specific grant from funding agencies in the public, commercial, or not-for-profit sectors.

\section{Conflict of Interest}

  The authors have no conflict of interest to declare.

\newpage

\appendix
\section{Partial derivative of inverse quadratic form}

In this section, it will be shown that the $i^{th}$ partial derivative of:
\begin{align}
  f(\bm{x}) & = \bm{v}^T (\bm{A}(\bm{x}))^{-1} \bm{v}
\end{align}
is
\begin{align}
  \frac{\partial f}{\partial x_i} & = -\bm{y}^T \frac{\partial \bm{A}}{\partial x_i} \bm{y}^T
\end{align}
where $\bm{A}$ is a matrix-valued function of $\bm{x}$, $\bm{v}$ is a constant vector and $\bm{y} = \bm{A}^{-1} \bm{v}$ is a an implicit function of $\bm{x}$ because $\bm{A}$ is a function of $\bm{x}$.

\begin{align}
  \bm{v} & = \bm{A} \bm{y} \\
  \bm{0} & = \bm{A} \frac{\partial \bm{y}}{\partial x_i} + \frac{\partial \bm{A}}{\partial x_i} \bm{y} \\
  \frac{\partial y}{\partial x_i} & = - \bm{A}^{-1} \frac{\partial \bm{A}}{\partial x_i} \bm{y} \\
  f(\bm{x}) & = \bm{v}^T \bm{A}^{-1} \bm{v} \\
  & = \bm{y}^T \bm{A} \bm{y} \\
  \frac{\partial f}{\partial x_i} & = 2 \bm{y}^T \bm{A} \frac{\partial \bm{y}}{\partial x_i} + \bm{y}^T \frac{\partial \bm{A}}{\partial x_i} \bm{y} \\
  & = - 2 \bm{y}^T \bm{A} \bm{A}^{-1} \frac{\partial \bm{A}}{\partial x_i} \bm{y} + \bm{y}^T \frac{\partial \bm{A}}{\partial x_i} \bm{y} \\
  & = - 2 \bm{y}^T \frac{\partial \bm{A}}{\partial x_i} \bm{y} + \bm{y}^T \frac{\partial \bm{A}}{\partial x_i} \bm{y} \\
  & = - \bm{y}^T \frac{\partial \bm{A}}{\partial x_i} \bm{y}
\end{align}

\newpage


%
%

\bibliographystyle{spbasic}      
\bibliography{references}   

%
%

\end{document}

%% file: figures/cantbeam.tex
\begin{figure}
  \centering
  \resizebox{0.4\textwidth}{!}{
    \begin{tikzpicture}
        \draw[fill,color=gray!70] (0,-0.1) rectangle (-0.3,2.1);
        \node [align=center, body,line width=1.2pt,minimum height=2cm,minimum width=8cm,anchor=south west] (body1) at (0,0) {};
        \draw (body1.south east) -- ++(1.6,0) coordinate (D1) -- +(5pt,0);
        \draw (body1.north east) -- ++(1.6,0) coordinate (D2) -- +(5pt,0);
        \draw [dimen] (D1) -- (D2) node {40mm};

        \draw (body1.north west) -- ++(0,2) coordinate (D1) -- +(0,5pt);
        \draw (body1.north east) -- ++(0,2) coordinate (D2) -- +(0,5pt);
        \draw [dimen] (D1) -- (D2) node {160mm};

        \draw[->,line width=1pt] (8,1) -- (9,1) -- (9,0.2);
        \node (arrowhead) at (8.7,0.25) {$\bm{F}_1$};

        \draw (body1.north west) -- ++(0,1.5) coordinate (D1) -- +(0,5pt);
        \draw (4,2) -- ++(0,1.5) coordinate (D2) -- +(0,5pt);
        \draw [dimen] (D1) -- (D2) node {80mm};

        \draw[->,line width=1pt] (3.4,2.6) -- (4,2);
        \node (arrowhead) at (3.2,2.8) {$\bm{F}_2$};

        \draw (body1.south west) -- ++(0,-1.5) coordinate (D1) -- +(0,5pt);
        \draw (6,0) -- ++(0,-1.5) coordinate (D2) -- +(0,5pt);
        \draw [dimen] (D1) -- (D2) node {120mm};

        \draw[->,line width=1pt] (6,0) -- (5.4,-0.6);
        \node (arrowhead) at (5.2,-0.8) {$\bm{F}_3$};

        \draw[->,line width=1pt] (0.3,2.3) -- (0.1,2);
        \draw[->,line width=1pt] (0.6,2.3) -- (0.8,2);
        \draw[->,line width=1pt] (1.2,2.3) -- (1.1,2);
        \draw[->,line width=1pt] (1.5,2.3) -- (1.8,2);
        \draw[->,line width=1pt] (2.2,2.3) -- (2.0,2);
        \draw[->,line width=1pt] (2.5,2.3) -- (2.8,2);
        \draw[->,line width=1pt] (3.1,2.3) -- (3.0,2);
        \draw[->,line width=1pt] (3.3,2.3) -- (3.5,2);
        \draw[->,line width=1pt] (4.5,2.3) -- (4.4,2);
        \draw[->,line width=1pt] (4.8,2.3) -- (4.9,2);
        \draw[->,line width=1pt] (5.3,2.3) -- (5.1,2);
        \draw[->,line width=1pt] (5.6,2.3) -- (5.9,2);
        \draw[->,line width=1pt] (6.2,2.3) -- (6.0,2);
        \draw[->,line width=1pt] (6.5,2.3) -- (6.9,2);
        \draw[->,line width=1pt] (7.2,2.3) -- (7.2,2);
        \draw[->,line width=1pt] (7.5,2.3) -- (7.7,2);
        \draw[->,line width=1pt] (8.1,2.3) -- (7.9,2);

        \draw[->,line width=1pt] (0.3,-0.3) -- (0.3,0);
        \draw[->,line width=1pt] (0.8,-0.3) -- (0.6,0);
        \draw[->,line width=1pt] (1.2,-0.3) -- (1.1,0);
        \draw[->,line width=1pt] (1.5,-0.3) -- (1.5,0);
        \draw[->,line width=1pt] (2.0,-0.3) -- (2.2,0);
        \draw[->,line width=1pt] (2.5,-0.3) -- (2.8,0);
        \draw[->,line width=1pt] (3.1,-0.3) -- (3.0,0);
        \draw[->,line width=1pt] (3.5,-0.3) -- (3.2,0);
        \draw[->,line width=1pt] (3.8,-0.3) -- (3.7,0);
        \draw[->,line width=1pt] (4.1,-0.3) -- (4.1,0);
        \draw[->,line width=1pt] (4.4,-0.3) -- (4.4,0);
        \draw[->,line width=1pt] (5.0,-0.3) -- (4.9,0);
        \draw[->,line width=1pt] (5.3,-0.3) -- (5.3,0);
        \draw[->,line width=1pt] (5.5,-0.3) -- (5.6,0);
        \draw[->,line width=1pt] (6.0,-0.3) -- (6.2,0);
        \draw[->,line width=1pt] (6.6,-0.3) -- (6.7,0);
        \draw[->,line width=1pt] (7.3,-0.3) -- (7.2,0);
        \draw[->,line width=1pt] (7.7,-0.3) -- (7.5,0);
        \draw[->,line width=1pt] (8.1,-0.3) -- (7.9,0);

        \draw[->,line width=1pt] (8.3,0.3) -- (8.0,0.1);
        \draw[->,line width=1pt] (8.3,0.5) -- (8.0,0.8);
        \draw[->,line width=1pt] (8.3,1.0) -- (8.0,1.2);
        \draw[->,line width=1pt] (8.3,1.5) -- (8.0,1.5);
        \draw[->,line width=1pt] (8.3,1.7) -- (8.0,1.9);
    \end{tikzpicture} \newline
  }
  \caption{Cantilever beam problem. $\bm{F}_2$ and $\bm{F}_3$ are at 45 degree angles.}
  \label{fig:CantBeam}
\end{figure}